\documentclass{article}
\usepackage{graphicx} 

\usepackage{PRIMEarxiv}

\usepackage[utf8]{inputenc} 
\usepackage[T1]{fontenc}    
\usepackage{url}            
\usepackage{booktabs}       
\usepackage{amsfonts}       
\usepackage{nicefrac}       
\usepackage{microtype}      
\usepackage{lipsum}
\usepackage{fancyhdr}       
\usepackage{graphicx}       
\graphicspath{{media/}}     
\usepackage{amsmath}
\usepackage{algorithm}
\usepackage{algorithmicx}
\usepackage{algpseudocode}
\usepackage{multirow}
\usepackage{todonotes}
\usepackage{comment}
\usepackage{caption}
\usepackage{subcaption}
\usepackage{cite}
\usepackage{amsthm}
\newtheorem{theorem}{Theorem}

\newtheorem{definition}{Definition}

\usepackage{enumitem}
\setlist[enumerate]{itemsep=0.1\normalbaselineskip}
\usepackage{hyperref}       

\newcounter{subdefinition}[definition]
\renewcommand{\thesubdefinition}{\thedefinition.\arabic{subdefinition}}

\newcommand{\VASI}{verifiable autonomous system information exchanges}
\newcommand{\papertitle}{Hermes’ Seal: Zero-Knowledge Assurance for Autonomous Vehicle Communications}

\title{
    \papertitle
}

\author{
Munawar Hasan$^{1,2}$,
Apostol Vassilev$^{1}$,
Edward Griffor$^{1}$,
Thoshitha Gamage$^{3}$ \\
\\
$^{1}$National Institute of Standards and Technology, USA \\
$^{2}$Michigan Technological University, USA \\
$^{3}$Southern Illinois University Edwardsville, USA \\
\\
\texttt{\{munawar.hasan, apostol.vassilev, edward.griffor\}@nist.gov} \\
\texttt{munawarh@mtu.edu, tgamage@siue.edu}
}
\date{June 2025}

\begin{document}

\maketitle

\begin{abstract}
The application of zero-knowledge proofs (ZKPs) in autonomous systems is an emerging area of research, motivated by the growing need for regulatory compliance, transparent auditing, and trustworthy operation in decentralized environments. zk-SNARK is a powerful cryptographic tool that allows a party (the prover) to prove to another party (the verifier) that a statement about its own internal state is true, without revealing sensitive or proprietary data about that state. This paper proposes Hermes’ Seal: a zk-SNARK-based ZKP framework for enabling privacy-preserving, verifiable communication in vehicle-to-vehicle (V2V) and vehicle-to-infrastructure (V2I) networks. The framework allows autonomous systems to generate cryptographic proofs of perception and decision-related computations without revealing proprietary models, sensor data, or internal system states, thereby supporting interoperability across heterogeneous autonomous systems.

We present two real-world case studies implemented and empirically evaluated within our framework, demonstrating a step toward \textbf{\VASI}. The first demonstrates real-time proof generation and verification, achieving $8$ ms proof generation and $1$ ms verification on a GPU, while the second evaluates the performance of an autonomous vehicle perception stack, enabling proof of computation without exposing proprietary or confidential data. Furthermore, the framework can be integrated into AV perception stacks to facilitate verifiable interoperability and privacy-preserving cooperative perception. The demonstration code for this project is open source, available on Github\footnote{\url{https://github.com/mhasan08/zk-AV}}.
\end{abstract}

\keywords{Autonomous vehicles \and zk-SNARKs \and Zero knowledge proofs \and Cooperative perception}

\section{Introduction}
\label{sec:introduction}
Autonomous Vehicles (AVs) have emerged as a transformative technology in recent years, promising to revolutionize transportation by improving road safety, reducing human error, and increasing efficiency in urban and rural mobility. The perception stack serves as the backbone of AVs, providing the foundation for situational awareness and downstream decision-making. This stack utilizes artificial intelligence (AI) models to process raw sensor data and extract meaningful structured information about the driving environment. AV manufacturers select modalities --- such as cameras, LiDAR, radar, or a combination thereof -- to achieve this goal. For instance, Tesla~\cite{tesla} uses a camera-centric approach, while Waymo~\cite{waymo} relies on a fusion of cameras, LiDAR, and high-definition (HD) maps. 

Whereas cameras offer rich semantic information, LiDAR provides superior performance in low-light conditions. Ultimately, the selection of a sensor suite represents a strategic trade-off between cost, computational complexity, and operational redundancy. Regardless of the chosen sensor modalities, standalone vehicle perception remains limited by physical constraints such as occlusions and limited line-of-sight in complex traffic scenarios -- a vulnerability shared by human drivers. Consequently, hidden obstacles and unusual traffic patterns can result in late detections that pose challenges to the vehicle's decision-making system, leading to potentially hazardous outcomes. This is why a vehicle, whether autonomous or human-driven, can benefit from  information from other roadway participants.  

\textbf{Cooperative perception}~\cite{coop-perception, liu2025collaborative, caillot2022survey, kim2015impact} enables multiple agents to enhance their situational awareness by sharing perception information about the driving environment, such as detected objects, road conditions, or safety-relevant events. In autonomous vehicle deployments, this information is often disseminated using a broadcast communication setting where an \textit{ego} vehicle transmits perception packets to any nearby recipients, such as other vehicles, roadside units (RSUs), infrastructure nodes, or even third-party auditors, without knowing their identities, affiliations, or establishing any prior relationships. Through such broadcasts, vehicles can extend their effective sensing horizon beyond line-of-sight limitations, thereby improving the detection of occluded objects, safety-critical events, and dynamic hazards.

To illustrate this, consider an \textit{ego} vehicle approaching an urban intersection surrounded by buildings on all four corners (see Figure~\ref{fig:traffic-pattern}). The \textit{ego} vehicle may be unable to detect a bicyclist approaching from a cross street due to occlusion by a building, whereas another nearby vehicle may observe the bicyclist and broadcast this information. Importantly, this transmission is intentionally  receiver-agnostic: the sender neither establishes pairwise relationships nor negotiates credentials with individual recipients. Perception information is made available to any entity within communication range, reflecting the highly dynamic and heterogeneous nature of vehicular environments.

However, this receiver-agnostic setup introduces a fundamental challenge for any receiving vehicle: how can externally received perception data be validated for correctness and policy compliance before being incorporated into safety-critical decision-making? Without a mechanism for independent verification, broadcast information may be incomplete, inconsistent, or adversarially manipulated.

To ensure safety and interoperability in this broadcast communication setting, a regulatory body like NHTSA (National Highway Traffic Safety Administration)~\cite{nhtsa} or a standards body like SAE (Society of Automotive Engineers)~\cite{sae} or a consumer safety organization like AAA (American Automobile Association)~\cite{aaa} in collaboration with the insurance industry defines a publicly specified policy that formalizes the minimum safety and performance constraints required for compliant perception reporting (e.g., confidence thresholds, geometric consistency, safety-critical object inclusion, etc.). Automotive manufacturers may voluntarily adopt this standard and configure their systems to generate perception outputs consistent with these constraints.

\begin{figure}
    \centering
    \includegraphics[width=0.82\linewidth]{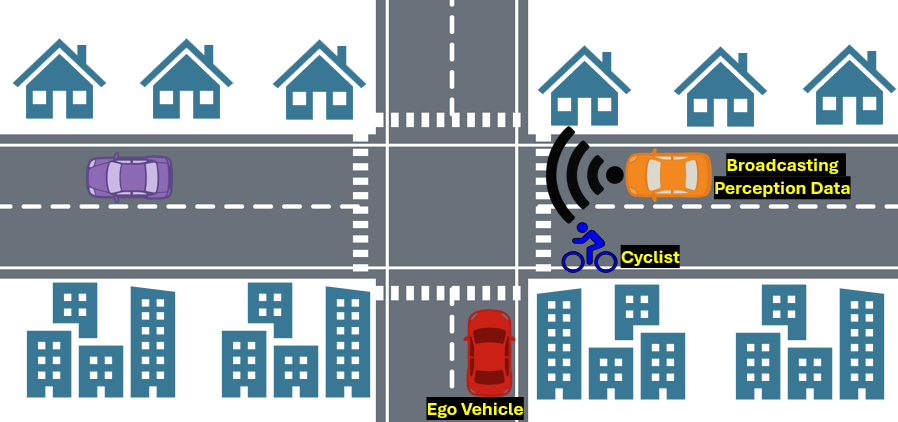}
    \caption{City traffic pattern: The perception stack of the \textit{Ego} vehicle (red car) is not enough to assess the traffic situation fully, since the cyclist is completely occluded by the city buildings. Hence, the \textit{Ego} vehicle needs perception data from another vehicle (orange car) to fully comprehend the situation and make an informed decision.}
    \label{fig:traffic-pattern}
\end{figure}

Under such a communication model, concerns about data privacy and integrity make it undesirable to disclose raw sensor data or proprietary data produced by AI models. Instead, a zero-knowledge framework can be employed whereby the ego vehicle attaches a cryptographic proof to each broadcast packet, certifying that the transmitted information satisfies the formally specified public policy. Consequently, any recipient, whether previously known or entirely unknown, can verify compliance without accessing sensitive internal details. Hence, trust is shifted from the identities of communicating participants to the mathematical validity of the proof of conformance to the public specification, thus enabling scalable, privacy-preserving cooperative perception in heterogeneous vehicular networks. Further, autonomous systems operate under extreme low-latency requirements. Accordingly, our approach is guided by two design objectives: (i) minimizing the number of messages exchanged, and (ii) minimizing the size of each transmitted message.

\textbf{Contributions of this paper:} \textbf{1.} We introduce Hermes' Seal\footnote{Hermes: the messenger of the Greek gods and a figure associated with hermeneutics—the art of interpretation—symbolizes trusted message transmission across boundaries, reflecting our goal of verifiable yet privacy-preserving AV communication.},  a framework for \textbf{\VASI} which enables the exchange of perception data with cryptographically verifiable correctness while concealing proprietary model information and raw sensory inputs.  \textbf{2.} We show how this framework can be deployed and used to solve two important open problems in automotive transportation: \textbf{2.a)} we address a fundamental limitation of cooperative perception, namely the lack of verifiable guarantees on externally received safety-critical data. As autonomous fleets continue to scale and interact in increasingly dense and heterogeneous environments, the need for mechanisms that simultaneously ensure verifiability, privacy, and interoperability becomes both timely and critical~\cite{remeli2019towards, jackson2021certified, gao2023online}. \textbf{2.b)} we develop a ZKP-based mechanism for privacy-preserving verification of vehicle compliance by authorities. 

The rest of this paper is organized as follows. Section \ref{sec:related-work} provides a brief primer on ZKPs and an overview of relevant literature, highlighting the computational constraints imposed by autonomous vehicles when applying ZKP techniques. Section \ref{sec:preliminaries} introduces the mathematical preliminaries and notation. Section \ref{sec:methodology} describes the formal specifications and the internal architecture of Hermes' Seal. Section~\ref{sec:security-proof} provides security analysis of the framework. Section \ref{sec:sub:application} presents two case studies of the framework. 
Section~\ref{sec:experiments} presents experimental results and analysis. Section~\ref{sec:limitation} and~\ref{sec:future-work} discuss current limitations and future work, respectively. Finally, Section~\ref{sec:conclusion} gives concluding remarks.

\section{A Primer on Zero-Knowledge Proofs}
\label{sec:related-work}

A zero-knowledge proof (ZKP) is a cryptographic protocol that enables a party, known as the \textit{prover}, to prove the validity of a specific statement --- often called the \textit{instance} --- to a \textit{verifier} without revealing the secret information -- known as the \textit{witness} --- that attests its truth. This witness could be a password, a transaction key, or proprietary machine learning model parameters that satisfy specific public constraints. The concept was introduced by Goldwasser et al.~\cite{goldwasser2019knowledge}, who defined the foundational model of interactive zero-knowledge proofs.

Fundamentally, a ZKP system must satisfy three properties:
\begin{enumerate}
    \item \textbf{Completeness}: If the statement is \texttt{true} and the prover possesses a valid witness, the protocol guarantees that an honest verifier will be convinced of the statement's validity.
    \item \textbf{Soundness}: If the statement is \texttt{false}, no dishonest prover can convince the verifier otherwise, except with negligible probability.
    \item \textbf{Zero-knowledge}: The proof reveals nothing beyond the validity of the statement; i.e., the verifier gains no additional information regarding the witness.
\end{enumerate}

Over the last decade, ZKPs have gained significant traction due to their utility in blockchain-based applications~\cite{ben2014succinct, sasson2014zerocash, groth2016size, zhou2024leveraging}. While ZKP initially sought for distributed computing and privacy-preserving transactions, interest has recently expanded into AI and machine learning (ML)~\cite{weng2021mystique, peng2025survey, aziz2024zkvml, hao2024scalable, wang2024efficient}. 

AI models deployed in commercial autonomous vehicles are resource-intensive, demanding substantial computational power and memory while operating under stringent latency constraints. Consequently, generating a ZKP for the internal state of an entire model is computationally prohibitive. For instance, zkCNN~\cite{liu2021zkcnn} demonstrates the formulation of convolution neural network (CNN) within a zero-knowledge framework; however, the authors report that generating a proof for the VGG16 model (configuration D)~\cite{simonyan2014very} requires $20$-$30$ seconds (s) and produces a proof of size $200$-$300$ kilobytes (KB). Similarly, Zkml~\cite{chen2024zkml}, a framework for ML inference, takes $52.9$ s to generate a $12$ KB proof. Such latency is unacceptable for real-time AV applications. Therefore, a framework for AVs must prioritize efficiency without sacrificing security. Additionally, in cooperative perception settings, messages are exchanged asynchronously over communication channels that may not support multi-round protocols. This raises an important design question: should the system rely on interactive or non-interactive proofs?

\subsection{Interactive vs Non-Interactive Proof}
\label{sec:sub:i-vs-non-i}
Zero-knowledge proofs are broadly classified into interactive and non-interactive categories based on the communication pattern between the prover and verifier. In interactive ZKPs, the verifier issues challenges to which the prover must respond consistently to validate the statement. While powerful, this requires a persistent communication channel and synchronized interaction.

In contrast, Non-Interactive Zero-Knowledge proofs (NIZKPs) eliminate the need for back-and-forth messaging. The prover generates a single, standalone proof that can be independently verified. NIZKPs, particularly in the form of zk-SNARKs (Zero-Knowledge Succinct Non-interactive Arguments of Knowledge), are uniquely suited for decentralized, asynchronous systems like AVs where low latency and broadcast capability are critical. For example, in a cooperative perception scenario, an \textit{ego} vehicle often broadcasts data to multiple surrounding vehicles. Maintaining persistent, interactive sessions with every receiving vehicle is infeasible. Conversely, zk-SNARKs align with this broadcast paradigm, allowing each receiving vehicle to independently verify the sender's proof. In addition to this, zk-SNARKs provide succinctness, which ensures that the proof size and verification time are significantly smaller than the cost of executing the original computation. Given these advantages, we focus exclusively on non-interactive zk-SNARKs for the remainder of this work.

\subsection{Computational Constraints in AVs} \label{sec:sub:motivation}

While Zero-knowledge proofs have found extensive applicability in blockchain and privacy-preserving identity systems~\cite{zcash, miers2013zerocoin, naik2021sovrin, ben2018scalable, zkevm}, their integration into high-performance ML remains an open challenge. Existing literature on verifiable ML can be broadly categorized into communication-intensive interactive protocols and latency-sensitive decentralized frameworks. Neither paradigm effectively addresses the unique constraints of AVs: the former imposes prohibitive bandwidth demands on  Vehicle-to-Everything (V2X) channels, while the latter introduces infrastructure dependencies and consensus delays that are incompatible with real-time decision loops.

Early research, such as SafetyNets~\cite{ghodsi2017safetynets}, utilized the GKR protocol~\cite{goldwasser2015delegating}, to provide interactive proofs. However, GKR is limited to low-depth arithmetic circuits and lacks support for complex conditionals. To circumvent this, SafetyNets necessitates the use of quadratic activation functions. This architectural constraint severely limits a neural network's ability to generalize due to the U-shaped curve of the activation function and can lead to unstable gradients during training.

Subsequent approaches like \texttt{Mystique}~\cite{weng2021mystique} target standard neural network inference but require interactive protocols unsuitable for AVs. As an example, for a ResNet50~\cite{he2016deep} model, \texttt{Mystique} requires a communication overhead exceeding $1.5 \text{ GB}$. As discussed before, such bandwidth demands are prohibitive for vehicle-to-everything (V2X) communication channels. Similarly, frameworks like vCNN~\cite{lee2020vcnn} allow for verifiable Convolutional Neural Networks (CNNs) but incur extreme latency, taking up to 8 h for proof generation and 19 s for verification, rendering them unusable for real-time safety systems. While more recent zero-knowledge proof systems can generate succinct proofs (often on the order of a few KBs) with prover complexity linear in the circuit size, these systems are primarily designed to support general-purpose computation or full-model verification. Hence, their applicability to real-time AV perception, which imposes stringent latency, bandwidth, and semantic constraints, remains an open challenge.

A parallel body of research explores blockchain-based verifiable computation~\cite{aziz2024zkvml, li2023aggregated, jain2021blockchain}. While these systems provide a decentralized trust root, they introduce significant infrastructure dependency. Blockchain-based solutions generally require distributed consensus mechanisms to validate state changes, introducing variable latency that is incompatible with the millisecond-level decision loops of AVs. Furthermore, the requirement for AVs to maintain continuous synchronization with a distributed ledger introduces unnecessary communication overhead. Unlike these approaches, our framework focuses on lightweight, point-to-point verifiability that functions independently of blockchain consensus layers.

\subsection{The zk-SNARK Protocol} \label{sec:sub:protocol}
\label{sec:sub:zksnark}
zk-SNARKs stands for \textbf{z}ero-\textbf{k}nowledge \textbf{S}uccinct \textbf{N}on-\textbf{I}nteractive \textbf{A}rguments of \textbf{K}nowledge. This is the core cryptograpic primitive used in the proposed Hermes' Seal framework. The term ``succinct" here implies that the proof size and verification time are independent or only polylogarithmically dependent on the size of
the underlying computation. Specifically, we employ the Groth16 protocol~\cite{groth2016size}, which is widely adopted in blockchain technologies~\cite{zcash} due to its efficiency. Groth16 produces proofs of constant size ($\approx 128$ bytes) and enables constant-time verification. In scope of this work, Groth16, however, comes at the cost of needing a trusted third party.

PLONK~\cite{plonk}, Halo2~\cite{halo2} and STARKs are few alternatives to Groth16 that avoid or eliminate a trusted setup requirement, but this comes at practical trade-off noteworthy for time-critical AV decision making.  Proofs in these systems are typically larger as compared to Groth16.
The resulting proof sizes and verification costs are generally less uniform and more dependent on circuit structure and implementation choices. Further, STARKs has higher communication overhead. This variability can pose challenges for time-critical AV deployments, especially in settings involving bandwidth-constrained V2X links, roadside verifiers with strict latency budgets, or resource- and thermal-constrained embedded platforms.

In contrast, Groth16 offers very small proofs and extremely fast deterministic with constant-time verification with constant-size inputs, making it well-suited for bandwidth-limited vehicle-to-infrastructure communication and on-vehicle real-time decision support, despite its reliance on a trusted setup. 

Conceptually, implementing a Groth16-based zk-SNARK involves a transformation pipeline. The computation (e.g., a neural network inference) is first expressed as a high-level \textbf{arithmetic circuit}. This circuit is transformed into a Rank-1 Constraint System (R1CS), a system of equations defining the computation's logic. Finally, the R1CS is converted into a Quadratic Arithmetic Program (QAP), a polynomial form that enables the generation of efficient proofs where the resulting proof has  three group elements. The mathematics behind Groth16 is deferred to Appendix \ref{app:groth16}.

\section{Preliminaries and Notation} \label{sec:preliminaries}

Before detailing the proposed framework, we formally define the notations and data representations used throughout this paper. Table~\ref{tab:abbreviations} summarize the abbreviations and mathematical notations used in this paper.

\begin{table}[H]
    \centering
    \caption{Symbols and Mathematical Notations}
    \label{tab:abbreviations}
    \begin{tabular}{|l|p{10cm}|}
        \hline
        \textbf{Symbol / Abbreviation} & \textbf{Description} \\
        \hline
        EA & Enforcing authority\\
        \hline
        $\mathcal{P}$ & Prover (e.g., the Ego Vehicle)\\
        \hline
        $\mathcal{V}$ & Verifier (e.g., Roadside Unit (RSU) or Auditor)\\
        \hline
        $\mathcal{L}$ & Set of constraints imposed by the EA \\
        \hline
        $\mathcal{X}$ & Set of public inputs (\textit{Instance}) \\
        \hline
        $\mathcal{W}$ & Set of private inputs (\textit{Witness}) \\
        \hline
        $\mathcal{D}$ & Complete input set where $\mathcal{D} \leftarrow (\mathcal{W} \cup \mathcal{X})$ \\
        \hline
        $\Delta$ & Four-byte Domain Separator\\
        \hline
        \texttt{Circuit} & Circuit that implements all the constraints from $\mathcal{L}$ using $\mathcal{X}$ and $\mathcal{W}$ \\
        \hline 
        R1CS & Rank-1 Constraint System \\
        \hline
        $\nu$ & Cryptographic Nonce \\
        \hline
        $T$ & Timestamp (Unix epoch in UTC) \\
        \hline
        $s_{sec}$ & Secret parameter used by $\mathcal{P}$ for commitment $c$, where $s_{sec} \in \mathbb{F}_p$ \\
        \hline
        $pk, vk$ & Proving key and Verification key \\
        \hline
        $\pi$ & Zero-knowledge proof / zk-SNARK proof \\
        \hline
        $pk_{sig}, vk_{sig}$ & Digital Signature Signing/Verification keys\\
        \hline
        $Cert_{VID}$ & \textit{Ego} vehicle's certificate \\
        \hline
        $\Pi$ & Proof package \\
        \hline
        CS & Commitment scheme \\
        \hline
        $c$ & Commitment \\
        \hline
        $H(A\; ; B)$ & Hash function in the ZK-circuit (Field arithmetic), where $A, B \in \mathbb{F}_p$\\
        \hline
        $H(A \parallel B)$ & Standard Hash function (Bitwise), where $\parallel$ denotes concatenation\\
        \hline
        $\mathbb{F}_p$ & Finite field of prime order $p$\\
        \hline
        $\mathbb{N}, \mathbb{R}$ & Sets of natural and real numbers, respectively\\
        \hline
        Enc, Dec & Encoding and decoding functions\\
        \hline
    \end{tabular}
\end{table}

\subsection{Data Representation}\label{sec:sub:data-rep}

Standard AV perception stacks operate on floating-point numbers. However, zk-SNARK circuits operate over a finite field $\mathbb{F}_p$. To bridge this gap, we employ a quantization scheme.

\begin{definition}\label{def:scale}
\textbf{Scaling Function: } Let $x \in \mathbb{R}$ and $\rho \in \mathbb{N}$, then we define a scaling function $\mathcal{F}$ that takes a real number and a scaling factor $\rho$ and outputs an element of the finite field $\mathbb{F}_p$ .i.e., $\mathcal{F}: \mathbb{R} \times \mathbb{N} \mapsto \mathbb{F}_p$. Mathematically,

\begin{equation}\label{eq:scale}
    \mathcal{F}(x, \rho) = \lfloor x \cdot \rho \rfloor \pmod{p}
\end{equation}

Similarly, we define $\mathcal{F}^{-1}$ as a function that takes a scaling factor $\rho$ and an element of finite field $\mathbb{F}_p$ denoted by $z$ and outputs a real number .i.e., $\mathcal{F}^{-1}: \mathbb{F}_p \times \mathbb{N} \mapsto \mathbb{R}$. Mathematically, 

\begin{equation}\label{eq:inversescale}
    \mathcal{F}^{-1}(z, \rho) = \frac{1}{\rho}
    \begin{cases}
        z & \text{if } z\leq \frac{p-1}{2}\\
      z-p & \text{otherwise}
    \end{cases}
\end{equation}

From equation~\eqref{eq:scale} and~\eqref{eq:inversescale}, for any $x$ satisfying $|\lfloor x \rho \rfloor| < \frac{p}{2}$, we have we have $\mathcal{F}^{-1}(\mathcal{F}(x, \rho), \rho) = x$. The choice of $\rho$ controls the trade-off between numerical precision and the risk of modular overflow. Example: Let $x = 0.75$ and $\rho = 100$. Then $\mathcal{F}(x,\rho) = \lfloor 0.75 \cdot 100 \rfloor = 75 \in \mathbb{F}_p$, and $\mathcal{F}^{-1}(75,\rho) = 0.75$. Since $75 \ll p/2$, modular wraparound does not occur in practice.
\end{definition}

\begin{definition}\label{def:encoding}
\textbf{Encoding/Decoding Functions}: Let $\text{Enc()}$ be an encoding function that takes an arbitrary input (text, integer, float etc.) and its associated datatype tag
$\mathsf{dtype}$, and converts it into a canonical representation in little-endian order. Hence, we have:

\begin{equation}
    \text{Enc}_{\mathsf{dtype}}(\cdot) \mapsto \{0, 1\}^*, 
\end{equation}
such that $\mathsf{dtype}$ denotes a datatype tag which specifies a canonical serialization format, and  $\mathsf{dtype} \in \mathcal{T}$, where $$\mathcal{T} =
\{\texttt{CTX}, \texttt{R1CS}, \texttt{VK}, \texttt{CERT},
\texttt{PROOF}, \texttt{COMMIT}, \texttt{TS}, \texttt{NONCE}\}.$$

The corresponding decoding function Dec() takes a canonical representation as an input and its associated datatype tag
$\mathsf{dtype}$, and returns the output by parsing the canonical representation based on datatype tag. Hence,

\begin{equation}\label{eq:decoding}
    \text{Dec}_{\mathsf{dtype}}(\{0,1\}^*) \mapsto x.
\end{equation}
\end{definition}

\section{Hermes' Seal: A ZKP-based Framework for AVs} \label{sec:methodology}

\subsection{Design Philosophy}\label{sec:sub:philosophy}

The architecture of Hermes' Seal is governed by four core design principles:

\begin{itemize}
    \item \textbf{Model Agnosticism}: The framework operates as a cryptographic wrapper around the perception pipeline, verifying its outputs without depending on the underlying perception architecture. Whether the AV uses YOLO~\cite{redmon2016you}, DETR~\cite{carion2020end}, or PointPillars~\cite{lang2019pointpillars}, the output is standardized into proof statements. This abstraction enables interoperability across heterogeneous fleets without enforcing uniformity at the model or data modality level.
    
    \item \textbf{Modular Integration}: Designed for post-inference integration, the framework functions as a distinct module in the AV stack. It requires no retraining of existing models and ensures that real-time inference performance is preserved, adding only a marginal overhead for proof generation. 
    
    \item \textbf{Privacy-Preserving Collaboration}: By verifying assertions rather than sharing raw data, the system facilitates V2X collaboration without exposing proprietary model weights or sensitive sensor streams. This ensures accountability while protecting intellectual property and complying with privacy regulations.
    
    \item \textbf{Extensibility}: While this work focuses on perception, the framework is scalable to other components of the autonomy stack, such as planning and control. By formulating application-specific constraints within the proving circuit, the same zero-knowledge mechanism can attest to safety compliance, ethical decision logic, or policy adherence across diverse verification needs.
\end{itemize}

\subsection{System Model and Entities}\label{sec:sub:system-model-entities}

To formalize the framework, we categorize the data involved into two types:

\begin{enumerate}
\item \textbf{Public Data ($\mathcal{X}$):} Also referred to as the \textit{Instance}, this includes non-sensitive information visible to all parties and used for verification. Examples include the claim being proven (e.g., coordinates of a detected obstacle or distance from an obstacle) and the cryptographic hash of the model architecture.
\item \textbf{Private Data ($\mathcal{W}$):} Also referred to as the \textit{Witness}, this includes proprietary or confidential information known only to the prover such as raw sensor data, model weights, internal neural network activations.
\end{enumerate}

Furthermore, we define three key entities governing the proof generation and verification process:
 
\begin{enumerate}
    \item \textbf{Prover ($\mathcal{P}$):} An entity that performs a computation on private or proprietary data and generates a cryptographic proof attesting to the correctness of the output. In this framework, the AV (or its onboard perception system) acts as the prover, ensuring its perception stack (e.g., object detection and classification) has executed the model correctly without modification.
    \item \textbf{Verifier ($\mathcal{V}$):} An entity that receives the proof and efficiently validates it without re-executing the underlying computation. A verifier can be an external party such as another vehicle (in V2X scenarios), roadside infrastructure, or a regulatory auditor ensuring compliance with safety standards.
    \item \textbf{Enforcing Authority} (EA): The entity responsible for defining the formal constraints and logical conditions that the prover $\mathcal{P}$ must satisfy. In an open, receiver-agnostic communication model, the EA provides the publicly specified policy against which transmitted perception packets are \textit{cryptographically verified}, i.e., with regard to a ZKP. It specifies application-level safety and performance requirements—such as minimum detection confidence thresholds, geometric consistency conditions, or mandatory inclusion of safety-critical objects—under which perception outputs are considered compliant.
    
    There are two possible trust configurations for the EA. In a \textit{trusted} configuration, the EA is assumed to act honestly and provide correct constraint specifications, as would be expected from a regulatory body (e.g., NHTSA)~\cite{nhtsa}), standards body (e.g., SAE~\cite{sae}) or a consumer safety organization (e.g., AAA~\cite{aaa}) or even an insurance provider seeking auditable safety guarantees. In a \textit{transparent} configuration, the EA’s role is limited to defining constraints, which are publicly verifiable, and the EA does not need to be trusted. In this work, we focus on the trusted EA model. The Hermes’ Seal framework is agnostic to the specific proof system employed. Consequently, extending the framework to support transparent proof systems or adversarial EAs remains an important direction for future work.
\end{enumerate}

\begin{figure}[t]
\caption{Sequence diagram of Hermes' Seal.}
\label{fig:sequencediagram}
\begin{center}
    \includegraphics[width=0.75\textwidth]{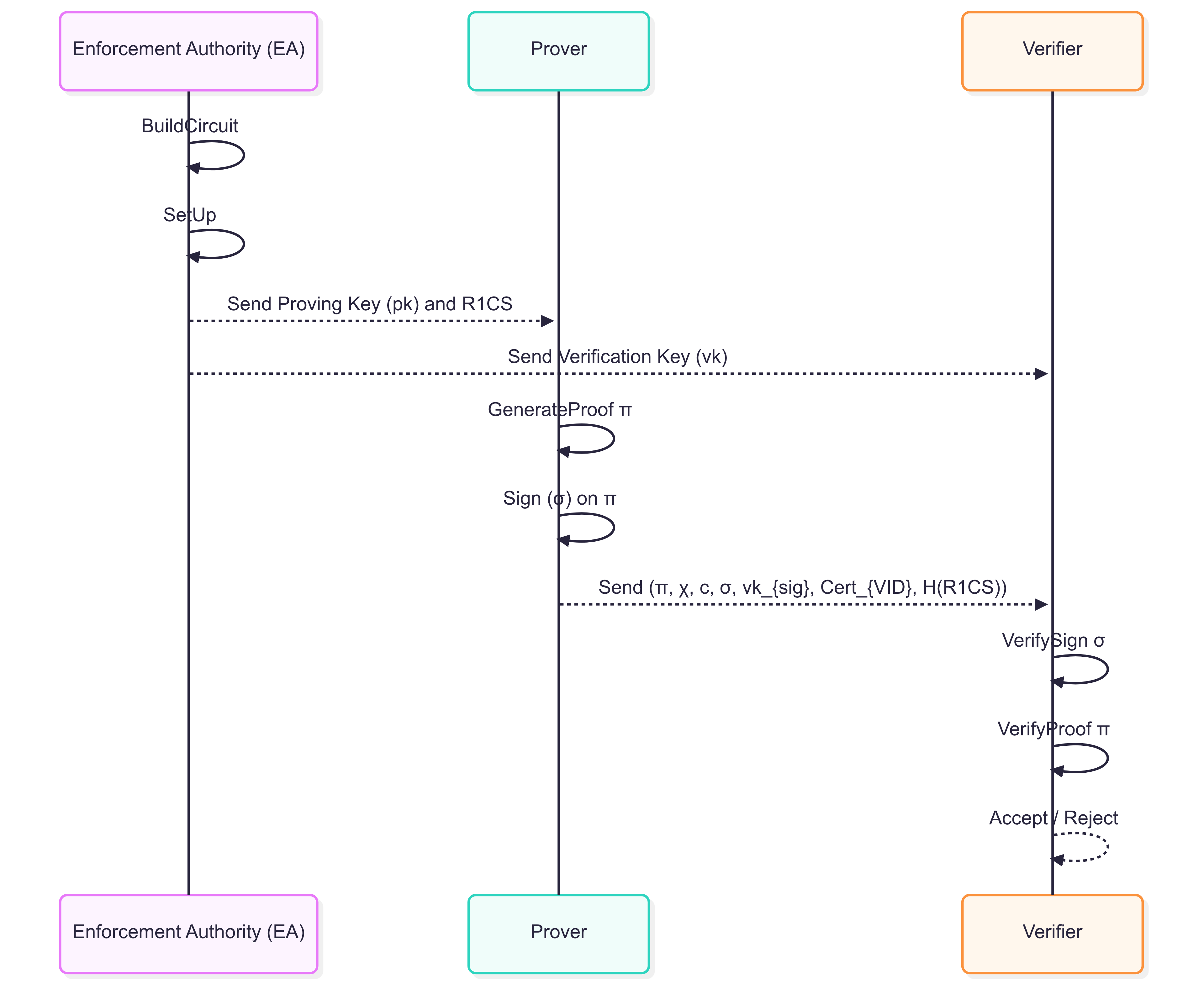}
\end{center}
\end{figure}

\subsection{Entity's Roles and Responsibilities}\label{sec:sub:entity-responsibility}

Figure~\ref{fig:sequencediagram} depicts the sequence diagram of Hermes' Seal, which illustrates the interactions between the three entities. The following descriptions detail their formal roles and responsibilities. It is important to note that, since ZKPs are tightly coupled with the applications in which they are deployed, the implementation details of the \texttt{Circuit} are presented here in abstract form. The concrete realization \texttt{Circuit} is discussed later in Section~\ref{sec:sub:application}, where two specific applications are examined in detail.

\begin{enumerate}
    \item Enforcement Authority (EA): EA serves as an anchor of trust and must be an organization possessing regulatory powers (e.g., NHTSA~\cite{nhtsa}). The EA has the following specific tasks and responsibilities.
    \begin{enumerate}
        \item \textit{Constraint formulation}: The EA defines the set of logical constraints or rules, denoted by $\mathcal{L}$, that the prover $\mathcal{P}$ must satisfy while generating the proof. For example, a constraint may specify that an AV can initiate a braking action only if the probability of detecting a relevant object exceeds $75\text{ \%}$.
        \item \textit{ZKP} \texttt{Circuit} \textit{Creation}: The EA uses $\mathcal{L}$ to design \texttt{Circuit}, which contains placeholders for both public inputs ($\mathcal{X}$) and private witnesses ($\mathcal{W}$), as well as commitment calculations ($c$). For instance, the actual probability of object detection is treated as a private input, while the object label is public. Parameters such as the scaling factor $\rho$ (refer to Definition~\ref{def:scale}) are also hard-coded into the \texttt{Circuit}. This work uses Circom~\cite{circom}, a domain-specific language and compiler for constructing arithmetic circuit. 
        \item \texttt{Circuit} \textit{Initialization}: The EA compiles the \texttt{Circuit} using Circom to generate a binary file known as R1CS (system of equations). Using the R1CS and a universal setup file called Powers of Tau (\texttt{ptau})\footnote{Powers of Tau: A universal parameter setup used in zk-SNARK systems to initialize public cryptographic parameters independent of any specific circuit, allowing the same setup to be reused across multiple circuits.}, the EA produces the proving key ($pk$) and the verification key ($vk$). $pk$ is distributed to authorized provers to enable proof generation, while $vk$ is made publicly available to all verifiers for proof validation.
        \item \textit{Auditing}: In post-incident scenarios or routine checks, the EA possesses the authority to demand the opening of specific commitments. The EA runs the Verify Algorithm (see Definition \ref{def:commitment-scheme}) to verify that the private inputs $m$ and randomness $s_{sec}$ match the committed values on the logs.
    \end{enumerate}
    \item \textbf{Prover ($\mathcal{P}$):} The Prover represents the autonomous vehicle (AV) or perception system responsible for generating a zero-knowledge proof ($\pi$). The prover executes two primary tasks:
    \begin{enumerate}
        \item \textit{Witness Generation:} The prover computes the witness values satisfying the circuit constraints. This involves combining the private data (e.g., scaled probability $\{75\} \in \mathcal{W}$) and public data (e.g., object label $\{\text{label}\} \in \mathcal{X}$) into a witness file formatted for the R1CS.
        \item \textit{Proof Generation \& Signing:} $\mathcal{P}$ uses the proving key $pk$ and the generated witness to compute the zk-SNARK ($\pi$)\footnote{Unless stated otherwise, we use \textit{proof} to denote the zero-knowledge proof artifact $\pi$, which serves as a proof of computation.}. Additionally, $\mathcal{P}$ signs the public inputs to ensure message authenticity.
    \end{enumerate}
    
    \item \textbf{Verifier ($\mathcal{V}$):} The verifier represents the entity (e.g., another AV or RSU) that verifies the proof correctness. The verification process is twofold: first, $\mathcal{V}$ verifies the digital signature to authenticate the source of the message (also identity of $\mathcal{V}$ is revealed). Once authenticated, $\mathcal{V}$ validates the cryptographic proof $\pi$ using the verification key $vk$ and the public data $\mathcal{X}$. Successful verification provides assurance that the underlying inference was computed correctly and satisfies the constraints of the EA, without revealing the private data of the prover.
\end{enumerate}

\subsection{Threat Model and Assumptions}\label{sec:sub:threatmodel}

To rigorously evaluate the security of the proposed framework, we define the capabilities of the adversaries and the trust assumptions placed on system entities. We analyze the security of the framework in the Random Oracle Model (ROM).

\begin{enumerate}
    \item \textbf{Adversarial Model ($\mathcal{A}$):} Let $\mathcal{A}$ be a probabilistic polynomial-time (PPT) adversary. $\mathcal{A}$ may act in two capacities:
    \begin{itemize}
        \item \textbf{Dishonest Prover:} $\mathcal{A}$ controls the vehicle and attempts to falsify perception outputs (e.g., claiming to see an obstacle when none exists) while still producing a valid proof $\pi$, violating the \textit{commitment-binding} and \textit{ZKP-soundness} properties.
        \item \textbf{External Observer:} $\mathcal{A}$ eavesdrops on the V2X communication channels and attempts to infer private perception values (e.g., raw sensor inputs) from the public transcripts and proofs,  violating the \textit{commitment-hiding} and \textit{zero-knowledge} properties.
    \end{itemize}
    
    \item \textbf{Trust Assumptions:} 
    \begin{itemize}
        \item \textbf{Enforcement Authority (EA):} is a Trusted Third Party (TTP), assumed to correctly generate circuit parameters ($pk, vk$) and securely deletes the toxic waste randomness during setup. Furthermore, the EA is trusted to perform retrospective audits by demanding the opening of commitments (via the Verify Algorithm) to verify historical data against the logs.
        \item \textbf{Sensor Integrity:} The physical integrity of the sensors (e.g., protection against GPS spoofing or sensor blinding) is outside the scope of this cryptographic framework. We assume the private input $\mathcal{W}$ correctly reflects the sensor's digital reading.
    \end{itemize}
\end{enumerate}

\subsection{Scope: Inference-Level Verification}\label{sec:sub:scope}

To ensure practical feasibility, we define the operational boundary of this framework as \textbf{Inference-Level Verification}. We employ zk-SNARKs to verify the correctness of specific inference outputs relative to defined safety constraints, rather than verifying the entirety of the autonomous software stack. This restriction addresses the core challenge of applying zk-SNARKs to AVs: balancing security with latency. Since verifying end-to-end neural network processing imposes prohibitive computational overhead, a narrower scope is essential for maintaining real-time performance in safety-critical scenarios.


This defined scope facilitates the modular integration of the framework across heterogeneous perception stacks. Consequently, each vendor can independently produce verifiable outputs without exposing their proprietary model architectures or training data. This modularity not only preserves the confidentiality of sensitive Intellectual Property but also enhances scalability by focusing strictly on the inference validity, striking a practical balance: providing cryptographic guarantees for perception integrity and V2X collaboration while maintaining the responsiveness required for safe driving.

The proposed EA is conceptually aligned with national regulators such as NHTSA. The EA acts as the trusted third party responsible for generating or certifying the proving ($pk$) and verification ($vk$) keys. Since the Groth16 protocol requires a Structured Reference String (SRS) --- the creation of which involves sensitive \textit{``secret trapdoor parameters"}\footnote{Known as ``toxic waste". See Appendix~\ref{app:groth16:srs}} --- this setup phase represents a critical point of trust. By assigning this trust responsibility to the EA, our model aligns with existing regulatory precedent~\cite{dot-connected-vehicle, etsi, underwriters-laboratories}, where authorities issue cryptographic credentials for vehicle identification, emissions compliance, and V2X communication. The use of a trusted third party for certifying proving/verification keys establishes a transparent governance mechanism, ensuring that cooperative perception is rooted in publicly verifiable, regulator-controlled trust rather than opaque vendor claims.

Furthermore, because Groth16 keys are bound to specific circuit constraints, the EA maintains a registry of valid verification keys ($vk$) corresponding to approved model architectures. This ensures that while vendors may utilize different proprietary models, every proof is validated against a key that was trusted and issued by the regulator. The mechanisms for distributing and managing these public parameters (e.g., PKI) are orthogonal to the cryptographic protocol and are outside the scope of this work.

In summary, the EA-operated setup ensures the following design guarantees:
\begin{enumerate}
    \item No vehicle manufacturer gains unilateral influence over the generation of the proving keys.
    \item During the trusted setup phase, secret trapdoor parameters generated for constructing the SRS is securely disposed of under regulatory oversight.
    \item The resulting verification keys ($vk$) can be published as an open standard for universal verification, while the proving keys ($pk$) are provisioned to manufacturers to generate valid proofs.
\end{enumerate}




\subsection{The Protocol Workflow}\label{sec:sub:workflow}

The proposed framework operates through a well-defined protocol consisting of six distinct \textbf{operations}, formally defined in \eqref{eq:createcircuit} -- \eqref{eq:verifyproof}. These operations encompass the complete lifecycle of the trust model and are organized into three sequential phases executed by the EA, the Prover, and the Verifier, respectively.

\begin{enumerate}
    \item \textbf{Operation: BuildCircuit} \\
    An abstract operation executed by the EA to translate safety policies into a computational form.
    \begin{itemize}
        \item \textbf{Input:} A set of logical constraints and rules $\mathcal{L}$ to be enforced.
        \item \textbf{Output:} A high-level arithmetic circuit, denoted by \texttt{Circuit}, which encodes the constraints in $\mathcal{L}$ as algebraic relations suitable for zk-SNARK proving.
        \item \textbf{Formulation:}
        \begin{equation}\label{eq:createcircuit}
            \texttt{Circuit} \leftarrow \text{\textbf{BuildCircuit}}(\mathcal{L})
        \end{equation}
    \end{itemize}
    Section~\ref{sec:sub:application} exemplifies concrete instantiations of the \texttt{Circuit}.
    \item \textbf{Operation: Setup} \\
    Executed by the EA, this operation initializes the cryptographic environment. It compiles the circuit and generates the cryptographic parameters required by the prover and verifier.
    \begin{itemize}
        \item \textbf{Input:} The \texttt{Circuit} and a publicly available universal parameter file called Powers of Tau (\texttt{ptau}).
        \item \textbf{Output:} A binary file encoding the Rank-1 Constraint System (R1CS), the proving key ($pk$), and the verification key ($vk$).
        \item \textbf{Formulation:}
        \begin{equation}\label{eq:setup}
            \begin{split}
                \text{R1CS}, pk, vk \leftarrow \textbf{Setup}&(\texttt{Circuit}, \text{ \texttt{ptau}}): \\
                 \text{R1CS} & \leftarrow \text{Compile}(\texttt{Circuit})\\
                 pk & \leftarrow \text{GenerateProvingKey}(\text{R1CS}, \text{\texttt{ptau}})\\
                 vk & \leftarrow \text{GenerateVerificationKey}(pk)
            \end{split}
        \end{equation}
    \end{itemize}
\end{enumerate}

Following the setup phase, the EA distributes the R1CS, $pk$ and $vk$ to both $\mathcal{P}$ and $\mathcal{V}$\footnote{In Groth16-based~\cite{groth2016size} zk-SNARK instantiations, it is secure to share the proving key with the verifier, though not strictly required for verification}. Additionally, it is assumed that each prover $\mathcal{P}$ is provisioned with a digital certificate $\text{Cert}_{VID}$ and a digital signature key pair $(sk_{sig}, vk_{sig})$. These credentials establishes the prover's identity. While these may be provisioned during manufacturing or registration, the specific Public Key Infrastructure (PKI) mechanism lies beyond the scope of this paper.

\begin{enumerate}
    \setcounter{enumi}{2} 

    \item \textbf{Operation: GenerateProof} \\
    Executed by the Prover, this operation uses the cryptographic parameters from the Setup phase and the input data $\mathcal{D}$ to generate the zk-SNARK proof, denoted by $\pi$.
    \begin{itemize}
        \item \textbf{Input:} The proving key ($pk$), binary file containing R1CS encoding, and the set of all inputs $\mathcal{D}$ (where $\mathcal{D} \leftarrow \mathcal{W} \cup \mathcal{X}$ such that $\mathcal{W}$ and $\mathcal{X}$ are private and public inputs respectively).
        \item \textbf{Output:} The cryptographic proof $\pi$ and the set of public inputs $\mathcal{X}$.
        \item \textbf{Formulation:}
        \begin{equation}\label{eq:generateproof}
            \begin{split}
                \mathbf{w}, (\pi, \mathcal{X}) \leftarrow \textbf{GenerateProof}&(pk, \text{ R1CS}, \mathcal{D}):\\
                \mathbf{w} & \leftarrow \text{GenerateWitness}(\text{R1CS}, \mathcal{D}) \\
                (\pi, \mathcal{X}) & \leftarrow \text{zkProof}(pk, \mathbf{w})
            \end{split}
        \end{equation}
    \end{itemize}
    
    \item \textbf{Operation: Sign} \\
    Executed by the Prover, this operation cryptographically signs the proof to bind it to the prover's identity.
    \begin{itemize}
        \item \textbf{Input:} The secret signing key $sk_{sig}$ and the message payload $m_{sig}$.
        \item \textbf{Output:} A digital signature $\sigma$.
        \item \textbf{Formulation:}
        \begin{equation}\label{eq:sign}
            \begin{split}
                \sigma \leftarrow \textbf{Sign}&(sk_{sig}, m_{\text{sig}}):\\
                 \sigma & \leftarrow \text{PerformSignature}(sk_{sig}, H(m_{sig}))
            \end{split}
        \end{equation}
    \end{itemize}
\end{enumerate}

Before broadcasting the proof $\pi$, $\mathcal{P}$ constructs the message payload $m_{sig}$ by concatenating the canonically encoded representations of the system parameters. This binds the proof to the specific circuit, verification key, and time.

\begin{equation}
    \begin{split}
        m_{sig} \leftarrow \Big( & \text{Enc}_{\texttt{CTX}}(\Delta_{Sign}) \parallel H(\text{Enc}_{\texttt{R1CS}}(\text{R1CS})) \parallel H(\text{Enc}_{\texttt{VK}}(vk)) \parallel H(\text{Enc}_{\texttt{CERT}}(\text{Cert}_{VID})) \\
        & \qquad \parallel H(\text{Enc}_{\texttt{PROOF}}(\pi)) \parallel \text{Enc}_{\texttt{COMMIT}}(c) \parallel \text{Enc}_{ \texttt{TS}}(T) \parallel \text{Enc}_{\texttt{NONCE}}(\nu) \Big)
    \end{split}
\end{equation}

where $\Delta_{Sign}$ is a domain separator tag, $c$ is the commitment output by the circuit (part of $\mathcal{X}$), $T$ is the current Unix epoch timestamp (in UTC), and $\nu$ is a random nonce for replay protection. Note that the hash function $H(\cdot)$ used in the signature scheme relies on standard cryptographic primitives (e.g., SHA-256), whereas the commitment $c$ is computed inside the circuit using an arithmetic-friendly hash function (e.g., Poseidon) to minimize proof generation overhead.

Let $\Pi_{VID}$ denote the proof package. Then, the AV generating the proof broadcasts the following proof package for the verifier: $$\Pi_{VID} = \{\pi, \mathcal{X}, c, \sigma, vk_{sig}, \text{Cert}_{VID}, H(\text{Enc}_{\texttt{R1CS}}(\text{R1CS}))\}.$$ 
Note that $m_{sig}$ is not transmitted as part of the broadcast payload. Instead, the verifier deterministically reconstructs $m_{sig}$ from the received fields and protocol context. $\text{Cert}_{VID}$ denotes the digital certificate binding the vehicle identity to its signing key $vk_{sig}$, issued by a certification authority. The design and operation of the PKI  are outside the scope of this work.

\begin{enumerate}
    \setcounter{enumi}{4}
    
    \item \textbf{Operation: VerifySign} \\
    Executed by the Verifier, this operation checks the authenticity of the message source.
    \begin{itemize}
        \item \textbf{Input:} The signature verification key $vk_{sig}$, the received signature $\sigma$, and the locally reconstructed payload $m_{sig}$.
        \item \textbf{Output:} A boolean status (\texttt{accept} or \texttt{reject}).
        \item \textbf{Formulation:}
        \begin{equation}\label{eq:verifysign}
            \begin{split}
            (\texttt{accept} / \texttt{reject}) \leftarrow \textbf{VerifySign}&(vk_{sig}, \sigma, m_{sig}):\\
                (\texttt{accept} / \texttt{reject}) \leftarrow & \text{VerifySignature}(vk_{sig}, \sigma, H(m_{sig}))
            \end{split}
        \end{equation}
    \end{itemize}

    \item \textbf{Operation: VerifyProof} \\
    Executed by the Verifier, this is the core validation step that checks the cryptographic correctness of the inference.
    \begin{itemize}
        \item \textbf{Input:} The zk-SNARK verification key $vk$, the proof $\pi$, and the public inputs $\mathcal{X}$.
        \item \textbf{Output:} A boolean status (\texttt{accept} or \texttt{reject}).
        \item \textbf{Formulation:}
        \begin{equation}\label{eq:verifyproof}
            \begin{split}
                (\texttt{accept} / \texttt{reject}) \leftarrow\textbf{VerifyProof}&(vk, \pi, \mathcal{X}):\\
                 (\texttt{accept} / \texttt{reject}) & \leftarrow \text{zkVerify}(vk, \pi, \mathcal{X})
            \end{split}
        \end{equation}
    \end{itemize}
\end{enumerate}

Upon the completion of the \textbf{VerifyProof} operation, the verification logic dictates the downstream behavior of the AV. If the output is \texttt{accept}, the AV treats the received information (e.g., the obstacle coordinates in $\mathcal{X}$) as validated truth and integrates it into its path-planning module. Conversely, if the output is \texttt{reject} --- whether due to an invalid signature or a failed zero-knowledge check—the AV discards the message entirely to prevent potential spoofing or erroneous data injection. This binary trust model ensures that only computationally verified perception data influences the vehicle's control decisions.

\subsection{Domain Separation and Context Binding}\label{sec:domainsep}

In a cooperative perception framework, an AV may generate proofs for multiple distinct logic gates (e.g., pedestrian detection, traffic light classification) using the same cryptographic identity. Without explicit context binding, a valid proof and signature intended for one context could theoretically be replayed and accepted in another if the public input structures are isomorphic.

To mitigate this \textit{cross-context ambiguity}, we enforce a strict domain separation schema. The framework employs a four-byte domain separator ($\Delta$) to uniquely identify the application scope and operation type. This value is prefixed to the message payload $m_{sig}$ before hashing, ensuring that the resulting signature $\sigma$ is mathematically bound to a specific context.

The $\Delta$ is constructed as follows:
\begin{itemize}
    \item \textbf{Byte 1 (App ID):} Identifies the specific AV application (e.g., \texttt{0x01} for Object Detection).
    \item \textbf{Byte 2 (Op ID):} Distinguishes between commitment generation (\texttt{0x01}) and signing operations (\texttt{0x02}).
    \item \textbf{Bytes 3-4 (Counter):} Reserved for versioning or indexing the total number of commitments associated with the application.
\end{itemize}

This structure ensures that a commitment generated inside the circuit cannot be confused with a signature generated outside of it, even if the underlying data is identical. The assignment of these values is performed by the EA during the design phase. As these are fixed parameters not consumed per proof, the identifier space is sufficient for practical scalability. (Specific domain separation values for the case studies discussed in this work are listed in Table~\ref{tab:flag} of Section~\ref{sec:sub:application}).

\section{Security Analysis}\label{sec:security-proof}

The proposed framework is evaluated against the threat model defined in Section~\ref{sec:sub:threatmodel}. This analysis first examines protocol-level defenses against active network attacks, followed by a formal cryptographic proof of the underlying Commitment Scheme (CS).

\subsection{Protocol Security Properties} \label{sec:sub:sec_prop}

Beyond the cryptographic primitives, the protocol workflow enforces operational security to mitigate active adversaries:

\begin{itemize}
    \item \textbf{Non-Repudiation}: The inclusion of the digital signature $\sigma$ in the proof package prevents the AV from denying authorship of a specific safety claim. Since the signing key $sk_{sig}$ is bound to the vehicle's hardware identity (via $\text{Cert}_{VID}$), any malicious message can be cryptographically traced back to the specific prover $\mathcal{P}$, satisfying the accountability requirement.
    \item \textbf{Replay Attack Protection}: To mitigate replay attack in which an adversary  captures a valid proof $\pi$ (e.g., "Brake applied") and rebroadcasts it at a later time, the framework enforces temporal freshness. The signature payload $m_{sig}$ includes a monotonically increasing Unix timestamp $T$ and a random nonce $\nu$. The Verifier $\mathcal{V}$ rejects any proof where $T$ is outside the acceptable latency window or where $\nu$ has been previously processed.\\
    We note that if the adversary controls the signing key (i.e., acts as a dishonest signer), timestamp-based freshness alone does not guarantee liveness, as such an adversary may precompute and sign messages using future timestamps. Preventing this class of attacks would require additional liveness guarantees, such as verifier-issued challenges, hardware-backed attestation or ZKP-based attestation~\cite{brandao2022zkasp}, which are outside the scope of this work.
    
    \item \textbf{Context Binding}: As detailed in Section \ref{sec:domainsep}, the use of the domain separator $\Delta$ ensures that proofs generated for one application logic (e.g., Object Detection) cannot be accepted by a Verifier expecting a different logic (e.g., Lane Keeping), even if the underlying zk-SNARK proof format is identical. Additionally, $H(R1CS)$, in the proof package $\Pi_{\mathit{VID}}$ provides explicit context binding between the proof $\pi$ and the intended computation, preventing its reuse for a different or a modified R1CS.
\end{itemize}

\subsection{Formal Cryptographic Analysis} \label{sec:sub:sec_formal}

The confidentiality and integrity of the perception data rely on the robustness of the underlying Commitment Scheme (CS) (formally defined in definition~\ref{def:commitment-scheme} in  Appendix~\ref{appendix:security-game}). The CS must satisfy two fundamental properties: \textit{Binding} (preventing the prover from changing values after commitment) and \textit{Hiding} (preventing the verifier from learning the values).

Below, we formally prove these bounds. We assume the Commitment Scheme consists of the tuple $\text{CS} = (\text{Commit}, \text{Open}, \text{Verify})$ and utilizes a hash function $H$ modeled as a Random Oracle.

\subsubsection{Binding Security}
\begin{theorem}\label{th:CS-bind}
    Let $\mathcal{A}$ be an adversary defined in Algorithm~\ref{alg:bind-game} (see Appendix~\ref{appendix:security-game}) that attacks the binding property of the commitment scheme. If $\mathcal{A}$ can output two distinct tuples $\tau = \{(m_1, s_{sec_1}), (m_2, s_{sec_2})\}$ such that $(m_1, s_{sec_1}) \neq (m_2, s_{sec_2})$ but $Commit(m_1, s_{sec_1}) = Commit(m_2, s_{sec_2})$, then we can construct a collision-finding adversary $\mathcal{B}_\mathcal{A}$ against the hash function $H$. The advantage is bounded by:
    \begin{equation}\label{eq:bind-bounds}
        \text{Adv}^{\text{BIND}}_{\text{CS}}(\mathcal{A}) \leq \text{Adv}^{\text{Coll}}_{H}(\mathcal{B}_\mathcal{A}).
    \end{equation}
\end{theorem}

\begin{proof}
    We construct an algorithm $\mathcal{B}_\mathcal{A}$ (Algorithm~\ref{alg:cr-bind}) that uses $\mathcal{A}$ to find a collision in $H$.
    
    \begin{algorithm}[!h]
        \caption{Reduction Adversary $\mathcal{B}_\mathcal{A}$}
        \label{alg:cr-bind}
        \begin{algorithmic}[1]
            \State $\tau \leftarrow \phi$
            \State \textbf{Run} $\mathcal{A}$ to obtain collision tuples:
            \State \hspace{0.5cm} $\{(m_1, s_{sec_1}), (m_2, s_{sec_2})\} \xleftarrow{\$} \mathcal{A}$
            \State \textbf{Compute Commitments:}
            \State \hspace{0.5cm} $c_1 \leftarrow H(m_1 \, ; \, s_{sec_1})$, $c_2 \leftarrow H(m_2 \, ; \, s_{sec_2})$
            \State \textbf{Check Collision:}
            \State \hspace{0.5cm} \textbf{if} $(c_1 = c_2) \wedge ((m_1, s_{sec_1}) \neq (m_2, s_{sec_2}))$ \textbf{then}
            \State \hspace{1.0cm} \textbf{return} $\tau = \{(m_1, s_{sec_1}), (m_2, s_{sec_2})\}$
            \State \hspace{0.5cm} \textbf{else}
            \State \hspace{1.0cm} \textbf{return} $\perp$
        \end{algorithmic}
    \end{algorithm}

    \noindent \textbf{Case 1: $(m_1, s_{sec_1}) = (m_2, s_{sec_2})$.}
    Since $H$ is deterministic, identical inputs yield identical outputs. This is not a collision. $\text{Adv}^{\text{BIND}}_{\text{CS}}(\mathcal{A}) = 0$.
    
    \noindent \textbf{Case 2: $(m_1, s_{sec_1}) \neq (m_2, s_{sec_2})$.}
    If $\mathcal{A}$ succeeds in finding valid openings for the same commitment $c$, it implies $H(m_1 \, ; \, s_{sec_1}) = H(m_2 \, ; \, s_{sec_2})$. Thus, any success by $\mathcal{A}$ translates directly to a collision found by $\mathcal{B}_\mathcal{A}$. Therefore, $\text{Adv}^{\text{BIND}}_{\text{CS}}(\mathcal{A}) = \text{Adv}^{\text{Coll}}_{H}(\mathcal{B}_\mathcal{A})$.

     Hence, from \textbf{Case 1} and \textbf{Case 2}, we conclude, $\text{Adv}^{\text{BIND}}_{\text{CS}}(\mathcal{A})$ $\leq \text{Adv}^{\text{Coll}}_{H}(\mathcal{B}_\mathcal{A})$.
\end{proof}

\subsubsection{Hiding Security}

\begin{theorem}\label{th:CS-hide}
    Let $\mathcal{A}$ be an adversary defined in Algorithm~\ref{alg:hide-game} (see Appendix~\ref{appendix:security-game}) that attempts to distinguish between the commitments of two messages $m_0$ and $m_1$. In the Random Oracle Model (ROM), the advantage of $\mathcal{A}$ is negligible:
    \begin{equation}\label{eq:hide-bounds}
        \text{Adv}^{\text{HIDE}}_{\text{CS}}(\mathcal{A}) = 0.
    \end{equation}
\end{theorem}

\begin{proof}
    The security relies on the properties of the Random Oracle $H$. For any message $M \in \mathcal{M}$, the probability of $H(M)$ outputting a specific value $y$ is uniformly distributed as $1/2^n$, where $n$ is the bit-length of the hash output. Crucially, because the adversary $\mathcal{A}$ does not possess the private randomness $s_{sec}$ used to salt the commitment, the hash output $c = H(m \parallel s_{sec})$ is statistically independent of the message $m$ from $\mathcal{A}$'s perspective.
    
    Formalizing this, we calculate the adversary's advantage defined as the probability of correctly guessing the bit $b$ (indicating which message was committed) minus the probability of a random guess ($\frac{1}{2}$). 
    
    From Algorithm~\ref{alg:hide-game}, we derive:
    \begin{equation}
        \begin{split}
             \text{Adv}^{\text{HIDE}}_{\text{CS}}(\mathcal{A}) &= \left| \Pr[\mathbf{G}_{\text{CS}}^{\text{HIDE}}(\mathcal{A}) \mapsto 1] - \frac{1}{2} \right| \\
             &= \left| \Pr[b = b'] - \frac{1}{2} \right| \\
             &= \Pr[b=1] \cdot \Pr[1 \xleftarrow{\$} \mathcal{A}(\tau_0, \tau_1, state, H(m_1 \parallel s_{sec_1}))] \\
             &\qquad + \Pr[b=0] \cdot \Pr[0 \xleftarrow{\$} \mathcal{A}(\tau_0, \tau_1, state, H(m_0 \parallel s_{sec_0}))] - \frac{1}{2}
        \end{split}
    \end{equation}
    
    where $\tau_i = (m_i, s_{sec_i})$. Since $H$ acts as a random oracle, the distribution of commitments $c$ is uniform over $\{0, 1\}^n$ regardless of the input message. We can marginalize over all possible oracle outputs $y$:
    
    \begin{equation}
        \begin{split}
             &= \frac{1}{2} \sum_{y \in \{0, 1\}^n} \Pr[c=y] \cdot \Pr[1 \xleftarrow{\$} \mathcal{A}(\dots, y)] \\
             &\qquad + \frac{1}{2} \sum_{y \in \{0, 1\}^n} \Pr[c=y] \cdot \Pr[0 \xleftarrow{\$} \mathcal{A}(\dots, y)] - \frac{1}{2}
        \end{split}
    \end{equation}
    
    Substituting the uniform probability $\Pr[c=y] = \frac{1}{2^n}$:
    
    \begin{equation}
        \begin{split}
             &= \frac{1}{2} \cdot \frac{1}{2^n} \sum_{y \in \{0, 1\}^n} \left( \Pr[1 \xleftarrow{\$} \mathcal{A}(\dots, y)] + \Pr[0 \xleftarrow{\$} \mathcal{A}(\dots, y)] \right) - \frac{1}{2}
        \end{split}
    \end{equation}
    
    Since $\mathcal{A}$ must output either 0 or 1, the sum of probabilities inside the parenthesis is exactly 1.
    
    \begin{equation}
        \begin{split}
             &= \frac{1}{2} \cdot \frac{1}{2^n} \cdot \sum_{y \in \{0, 1\}^n} (1) - \frac{1}{2} \\
             &= \frac{1}{2} \cdot \frac{1}{2^n} \cdot (2^n) - \frac{1}{2} \\
             &= \frac{1}{2} - \frac{1}{2} = 0
        \end{split}
    \end{equation}
    
    Thus, $\mathcal{A}$ gains no information about the committed message.
\end{proof}

\section{Applications to Autonomous Vehicles}\label{sec:sub:application}

Having established the abstract protocol and security properties in Section~\ref{sec:methodology}, we now demonstrate the practical utility of the framework through two concrete instantiations. These case studies illustrate how generic circuit construction can be adapted to different objectives -- safety verification and model auditing -- while preserving core privacy guaranties.

For clarity, we adopt standard autonomous driving terminology: the vehicle generating the proof is referred to as the \textit{ego} vehicle.
\begin{itemize}
    \item Case Study I (Section~\ref{sec:sub:perception-integrity}): Focuses on \textbf{Perception Integrity}. The ego vehicle proves it is maintaining a safe distance from a stop sign according to the RSS safety model~\cite{mobileye,vassilev2026assessment}, while simultaneously broadcasting the sign's location to trailing vehicles.
    \item Case Study II (Section~\ref{sec:sub:semantic}): Focuses on \textbf{Semantic Interoperability}. The ego vehicle proves the performance (Precision/Recall) of its perception stack on a standardized challenge set without revealing its proprietary model outputs.
\end{itemize}

To prevent cross-context replay attacks (as defined in Section~\ref{sec:domainsep}), each instantiation is assigned a unique Domain Separator ($\Delta$). These 4-byte tags are prefixed to the cryptographic commitments and signatures to bind them to their specific application logic. Table~\ref{tab:flag} details the specific values used for these case studies, formatted as [App ID] [Op ID] [Counter].

\begin{table}[!h]
    \centering
    \caption{Domain Separation Tags ($\Delta$) for Case Studies. Values are presented in Hexadecimal to illustrate the byte-level structure defined in Section~\ref{sec:domainsep}.}
    \label{tab:flag}
    \begin{tabular}{l c c}
    \toprule
    \textbf{AV Application} & \textbf{Commitment Tag} $(\Delta_{commit})$ & \textbf{Signature Tag} $(\Delta_{Sig})$ \\
    \midrule
    \textbf{Structure} & \texttt{0x[App][01][0000]} & \texttt{0x[App][02][0000]} \\
    \midrule
    \textbf{Case I:} Perception Integrity & \texttt{0x00010000} & \texttt{0x00020000} \\
    \small{(App ID: \texttt{0x00})} & \small{(Decimal: 65,536)} & \small{(Decimal: 131,072)} \\
    \midrule\textbf{Case II:} Model Audit & \texttt{0x01010000} & \texttt{0x01020000} \\
    \small{(App ID: \texttt{0x01})} & \small{(Decimal: 16,842,752)} & \small{(Decimal: 16,908,288)} \\
    \bottomrule
    \end{tabular}
\end{table}

\subsection{Case Study I: Proof of Perception Integrity}\label{sec:sub:perception-integrity}

To demonstrate the framework in a safety-critical context, we consider the case of \textit{cooperative perception} to check safe stopping distances. In this scenario, an ego vehicle generates a cryptographic proof attesting that it is maintaining a safe distance from a detected object (e.g., a stop sign) without revealing its exact speed, location, or raw sensor data.

This instantiation maps the abstract framework to a concrete problem: enabling an observer (such as a trailing vehicle or road infrastructure) to verify that the ego vehicle's perception system has correctly identified a safety constraint and is adhering to it.




\subsubsection{Safety Formulation: RSS Distance to Stop Sign} \label{sec:sub:safe-distance}
We adopt the longitudinal safe-distance formulation of the Responsibility-Sensitive Safety (RSS) model~\cite{mobileye}. RSS defines the minimum distance $d_{min}$ that an ego vehicle must maintain from a leading object to guarantee a collision-free stop, accounting for reaction times and braking capabilities.


The general RSS safe-distance expression is given by:

\begin{equation}\label{eq:mobileye-safe-distance}
    d_{min} = \left[ v_r t_{rec} + \frac{\alpha_{max} t_{rec}^2}{2} + \frac{(v_r + t_{rec} \alpha_{max})^2}{2\beta_{min}} - \frac{v_f^2}{2\beta_{max}} \right],
\end{equation}

where, $v_r$ and $v_f$ are the speeds of the rear (\textit{ego}) and front vehicles respectively, $t_{rec}$ denotes the response time of the rear (\textit{ego}) vehicle, $\alpha_{max}$ is the maximum acceleration during the reaction interval, and $\beta_{min}$, $\beta_{max}$ denote minimum and maximum braking capabilities. 

For this case study, we model the stop sign as a stationary object ($v_f = 0$). We assume the ego vehicle does not accelerate towards the stop sign ($\alpha_{max} = 0$). Let $v_r = v$ and the braking capacity $\beta_{min} = \mu g$, where $\mu \approx 0.75$ (friction coefficient for asphalt) and $g=9.81 \text{ m/s}^2$.

Under these assumptions, Equation~\eqref{eq:mobileye-safe-distance} reduces to the required safe distance $d_S$:

    \[ d_S = \left[ vt_{rec} + \frac{v^2}{2\mu g} \right] \]

For example, an \textit{ego} vehicle traveling at 30 mph ($\approx 13.41$ m/s) with a reaction time $t_{rec} = 1$ s, requires a reaction distance of $13.41$ m and a braking distance of $12.22$ m. Therefore, the vehicle must maintain a total safe buffer of $25.6$ m to ensure a safe stop.


\subsubsection{Framework Instantiation}

\textbf{Problem Statement:} The \textit{ego} vehicle (Prover $\mathcal{P}$) must broadcast a proof $\pi$ attesting that its current distance from the stop sign ($d^C_{\mathcal{S}}$) is greater than the calculated safe stopping distance ($d_S$). Mapping this problem to the framework defined in Section~\ref{sec:methodology}, the Enforcement Authority (EA) defines the logic $\mathcal{L}$ (Algorithm~\ref{alg:perceptionintegrity}) and compiles it into the \texttt{Circuit} (Algorithm~\ref{alg:zksnark-safetyquadrant}). The data is partitioned as follows:

\begin{enumerate}
    \item \textbf{Public Data ($\mathcal{X}$):} Inputs required for context and verification, including the stop sign location (for trailing vehicles) and the claimed safety status (\texttt{SAFE}).
    \item \textbf{Private Data ($\mathcal{W}$):} Sensitive internal states, including the vehicle's exact speed $v$, precise GPS location $(\varphi_{V}, \lambda_{V})$, and internal model probability $Pr$.
\end{enumerate}

\begin{algorithm}[!ht]
\caption{Safety Predicate $\mathcal{L}$}\label{alg:perceptionintegrity}
\begin{algorithmic}
        \If{$[Pr \geq \theta_{\mathcal{S}}] \wedge [d^{C}_{\mathcal{S}} \geq d_{\mathcal{S}}]$} \Comment{$d^{C}_{\mathcal{S}}$ denotes current distance of \textit{ego} vehicle from $\mathcal{S}$ and $d_S$ is the safe distance}
            \State \textbf{Return} $1$  \Comment{\texttt{SAFE}: \textit{Ego} vehicle is operating at a safe distance}
        \Else
            \State \textbf{Return} $0$ \Comment{\texttt{UNSAFE}: \textit{Ego} vehicle is not operating within a safe distance}
        \EndIf 
    \end{algorithmic}
\end{algorithm}

Table~\ref{tab:perceptionintegrity} details the parameters. Note that all real-valued inputs (GPS, probability) are mapped to the finite field $\mathbb{F}_p$ using the scaling function defined in Definition~\ref{def:scale}.

\begin{table}[!ht]
\centering
\caption{Circuit Parameters for Safe Distance Instantiation. All real-valued parameters are scaled to $\mathbb{F}_p$.}
\label{tab:perceptionintegrity}
\begin{tabular}{l p{9.5cm}}
    \toprule
    \textbf{Parameter} & \textbf{Description} \\
    \midrule
    \multicolumn{2}{c}{\textit{\textbf{Public Inputs (Instance $\mathcal{X}$): Known to Verifier}}} \\
    \midrule
    ID & Object Class ID (e.g., 11 for Stop Sign) \\
    $d_{\mathcal{S}}$ & Required safe stopping distance (calculated via RSS) \\
    $d^{C}_{\mathcal{S}}$ & Current distance from stop sign \\
    $\varphi_{\mathcal{S}}, \lambda_{\mathcal{S}}$ & Latitude and Longitude of the stop sign \\
    $\rho_{prob}, \rho_{geo}, \rho_{\psi}$ & Scaling factors for probability, geo-coordinates, and yaw \\
    $T, \nu$ & Unix timestamp and random nonce (for freshness) \\
    $c$ & Cryptographic commitment hash of private inputs \\
    \texttt{SAFE} & Boolean flag: 1 (Safe) or 0 (Unsafe) \\
    $w_{cloud/prep/fog}$ & Weather parameters (Cloud, Precipitation, Fog density) \\
    
    \midrule
    \multicolumn{2}{c}{\textit{\textbf{Private Inputs (Witness $\mathcal{W}$): Known only to Prover}}} \\
    \midrule
    $Pr$ & Detection probability score ($Pr \in [0, 1]$) \\
    $\mathbf{b}$ & Bounding box coordinates $[x_1, y_1, x_2, y_2]$ \\
    $\varphi_{V}, \lambda_{V}$ & Latitude and Longitude of the Ego vehicle \\
    $v, \psi$ & Speed and Yaw of Ego vehicle \\
    $s_{sec}$ & 128-bit secret blinding factor for commitment \\
    \bottomrule
\end{tabular}
\end{table}

\subsubsection{Circuit Construction}

Algorithm~\ref{alg:zksnark-safetyquadrant} details the \textbf{BuildCircuit} operation. Unlike a standard software function that executes on specific inputs at runtime, this operation defines the static \textbf{Rank-1 Constraint System (R1CS)} structure during the setup phase. It establishes the immutable arithmetic gates that constrain the relationship between the private witness (sensor data) and the public output (\texttt{SAFE}). 

This circuit enforces two critical properties:
\begin{enumerate}
    \item \textbf{Computational Correctness:} It encodes the RSS formula as arithmetic constraints, ensuring that a valid proof can only be generated if the distance threshold logic is satisfied.
    \item \textbf{Input-Output Binding:} It cryptographically binds the public output signal \texttt{SAFE} to the private inputs. This prevents a malicious prover from evaluating the safety logic as "Unsafe" internally but broadcasting ``Safe" publicly, as the resulting proof would fail verification against the R1CS.
\end{enumerate}

\begin{algorithm}[!ht]
  \caption{\textbf{Circuit Definition: Safety Predicate $\mathcal{L}$}} 
  \label{alg:zksnark-safetyquadrant}
  \textbf{Operation:} \textsc{BuildCircuit}$(\mathcal{L})$ \\
  \textbf{Description:} Defines the arithmetic constraints and signal wiring. This operation is executed once by the EA to generate the R1CS.
  
  \begin{algorithmic}[1]
    \State \textbf{Signals Declaration:}
    \State \quad $\text{signal input private } \{Pr, \mathbf{b}, \varphi_V, \lambda_V, v, \psi, s_{sec}\}$ \Comment{Witness $\mathcal{W}$}
    \State \quad $\text{signal input public } \{d_{\mathcal{S}}, d^C_{\mathcal{S}}, \text{ID}, \Delta_{commit}, \dots \}$ (refer Table~\ref{tab:perceptionintegrity}) \Comment{Instance $\mathcal{X}$ includes Domain Tag}
    \State \quad $\text{signal output public } \{\texttt{SAFE}\}$ \Comment{Result}

    \Statex \hrulefill
    
    \State \textbf{Constraint 1: Object Class Verification}
    \State \quad $\texttt{check\_id} \leftarrow \text{IsEqual}(\text{ID}, \text{ID}_{\mathcal{S}})$ 
    \State \quad $\text{Assert}(\texttt{check\_id} === 1)$ \Comment{Proof is only valid for Stop Signs}

    \State \textbf{Constraint 2: Threshold Logic (The Safety Gate)}
    \State \quad $\texttt{is\_detected} \leftarrow \text{GreaterThanOrEqual}(Pr, \theta_{\mathcal{S}})$
    \State \quad $\texttt{is\_distant} \leftarrow \text{GreaterThanOrEqual}(d^C_{\mathcal{S}}, d_{\mathcal{S}})$
    
    \State \quad \textit{// Logic: If detected, must be distant. If not detected, effectively safe.}
    \State \quad $\texttt{safety\_distance\_condition} \leftarrow \text{OR}(\text{NOT}(\texttt{is\_detected}), \texttt{is\_distant})$
    
    \State \quad $\texttt{SAFE} \ \text{<==} \ \texttt{safety\_distance\_condition}$ \Comment{Constrain Output Signal}

    \State \textbf{Constraint 3: Context Binding (Commitment)}
    \State \quad $\text{Assert}(\Delta_{commit} === \texttt{0x00010000})$ \Comment{Enforce correct App Context}
    \State \quad \textit{// Calculate commitment binding the context to the private predictions}
    \State \quad $c \leftarrow \text{PoseidonHash}(\Delta_{commit};Pr; \mathbf{b}; \varphi_V; \lambda_V; v; \psi; T; \nu; s_{sec})$ \Comment{Enforce Binding and Hiding}
    
  \end{algorithmic}
\end{algorithm}

\subsubsection{Usage Scenarios}
The instantiation defined in Algorithm~\ref{alg:zksnark-safetyquadrant} enables several practical capabilities for V2X systems:

\begin{figure}[!ht]
\caption{Operational Scenario: The ego vehicle (leading) detects a stop sign on a winding road. It broadcasts a zk-SNARK proof attesting to the stop sign's location and the vehicle's safe stopping distance. The trailing vehicle, whose view is occluded by vegetation, verifies this proof to update its local map and velocity planning.}
    \label{fig:perception-integrity-stop-sign}
    \centering
    \includegraphics[width=0.75\linewidth]{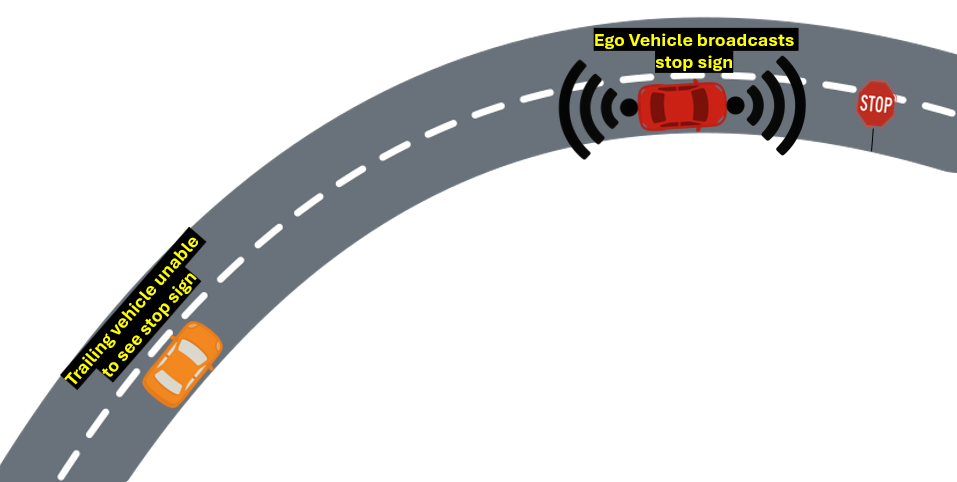}
\end{figure}

\begin{itemize}
    \item \textbf{Preemptive Situational Awareness:} As illustrated in figure~\ref{fig:perception-integrity-stop-sign}, a trailing vehicle (occluded by a curve or obstacle), can receive verified metadata about a stop sign ahead. Because the proof guarantees the integrity of the the data without revealing the leading vehicle's raw sensor stream, the trailing vehicle can confidently trust this ``\textit{Over-the-Horizon}" perception to adjust its own planning. 
    \item \textbf{Regulatory Compliance:} The \textit{ego} vehicle can mathematically prove to law enforcement or digital road infrastructure that its perception system is meeting the mandated performance bar (e.g., $Pr \geq 75\text{ }\%$) without submitting the proprietary model for inspection.  
    \item \textbf{Liability Auditing:} In post-accident forensics, the \textit{ego} vehicle can prove it was operating within the safe braking distance Equation~\ref{eq:mobileye-safe-distance}) at the time of the incident, effectively proving "absence of fault" regarding the longitudinal distance, while keeping the exact braking curve private.
    \item \textbf{Trustworthy Semantic Broadcast:} This mechanism transforms V2X from a naive data echange into a verified semantic broadcast. Consuming vehicles do not need to blindly trust the sender; they verify the zk-SNARK proof to ensure the semantic label (Stop Sign) and the safety state (Safe Distance) were generated with valid R1CS circuit. 
\end{itemize}

\subsection{Case Study II: Semantic Interoperability \& Model Auditing}\label{sec:sub:semantic}

While the previous case study focused on a specific safety rule (safe distance), this case study addresses the broader challenge of Semantic Interoperability and Performance Auditing.

In cooperative perception, a vehicle's situational awareness often depends on data received from others. To trust this external data, the receiving entity (or a regulator) requires assurance that the sender's perception stack performs reliably. However, proving model performance usually requires sharing the model's outputs on a test set, which risks leaking proprietary intellectual property.

We propose using the \textbf{\VASI} framework to solve this dilemma. The Enforcement Authority (EA) issues a ``Challenge Set" of images. The ego vehicle processes these images and generates a Zero-Knowledge Proof (ZKP) attesting that its model $M$ achieves specific Precision and Recall targets, without revealing the actual bounding boxes or class probabilities.


\subsubsection{Formal Verification of Object Detection}\label{sec:sub:detect_min_objects}

\textbf{Problem Statement:} The \textit{ego} vehicle (Prover) must prove that its perception model $M$ detects objects in a challenge dataset $\mathcal{I}$ with sufficient accuracy. Specifically, for each image, the model must satisfy a minimum Precision threshold ($\tau_{prec}$) and Recall threshold ($\tau_{rec}$). Additionally, the model must successfully detect all objects flagged as \textit{Safety Critical} (e.g., pedestrians, red lights).



Let $\mathcal{I} = \{I_1, \dots, I_N\}$ be the set of challenge images. For each image $I_i$, let the Ground Truth set be $\mathcal{G}_i = \{g_1, \dots, g_{K_i}\}$, , where $K_i$ is the number of objects in image $i$. Each ground truth object $g_k$ is defined by a bounding box $\mathbf{b}_k$ and a class label $id_k$. A subset of these objects may be marked as safety-critical, denoted by the indicator $c_k \in \{0, 1\}$.

Let the Model Detections for image $I_i$ be $\mathcal{D}_i = \{d_1, \dots, d_{M_i}\}$. Each detection $d_j$ consists of a predicted box $\mathbf{b}'_j$, class $id'_j$, and confidence score $Pr'_j$.

To verify performance, we employ the standard Greedy IoU Matching protocol used in benchmarks like COCO~\cite{coco}. A detection $d_j$ matches a ground truth $g_k$ if and only if:

\begin{enumerate}
    \item Class Match: $id'_j = id_k$
    \item Confidence Check: $Pr'_j \geq \theta_{conf}$ (Confidence Threshold)
    \item Overlap Check: $\text{IoU}(\mathbf{b}'_j, \mathbf{b}_k) \geq \theta_{IoU}$ (Intersection-over-Union Threshold)
\end{enumerate}

We define the set of True Positives (TP) as the count of unique greedy matches between $\mathcal{D}_i$ and $\mathcal{G}_i$. The system is considered Valid only if the following conditions hold for the dataset:

\begin{equation}
    \text{Precision: } \frac{|\text{TP}|}{|\mathcal{D}_i|} \geq \tau_{prec} \quad \wedge \quad \text{Recall: } \frac{|\text{TP}|}{|\mathcal{G}_i|} \geq \tau_{rec}
\end{equation}

Additionally, for the subset of safety-critical ground truths (where $c_k = 1$), the recall must be $100\text{ }\%$.

\subsubsection{Framework Instantiation}
Table~\ref{tab:semanticinteroperability} details the circuit parameters. The Challenge Images are handled via cryptographic commitment (hash) to ensure the vehicle is testing against the correct dataset without needing to embed the full images into the circuit.

\begin{table}[!ht]
    \centering
    \caption{Circuit Parameters for Model Performance Auditing.}\label{tab:semanticinteroperability}
    \begin{tabular}{l p{9.5cm}}
        \toprule
        \textbf{Parameter} & \textbf{Description} \\
        \midrule
        \multicolumn{2}{c}{\textit{\textbf{Public Inputs (Instance $\mathcal{X}$): Known to Verifier}}} \\
        \midrule
        $H(\mathcal{I})$ & Hash of the Challenge Dataset (ensures correct test set used) \\
        $\mathcal{G}$ & Ground Truth labels and boxes for the challenge set \\
        $\theta_{conf}, \theta_{IoU}$ & Minimum confidence and IoU thresholds (e.g., 0.5) \\
        $\tau_{prec}, \tau_{rec}$ & Required Precision and Recall pass marks \\
        $\rho_{prob}, \rho_{bbox}$ & Scaling factors for probabilities and coordinates \\
        $T, \nu$ & Unix timestamp and random nonce \\
        $c$ & Commitment to the private predictions \\
        \texttt{PASS} & Boolean flag: 1 (Model Passed) or 0 (Model Failed) \\
        \midrule
        \multicolumn{2}{c}{\textit{\textbf{Private Inputs (Witness $\mathcal{W}$): Known only to Prover}}} \\
        \midrule
        $\mathcal{D}$ & Model Predictions (List of $[id', Pr', x_1', y_1', x_2', y_2']$) \\
        $s_{sec}$ & 128-bit secret blinding factor \\
        VID & Vehicle Hardware ID (bound to the signature) \\
        \bottomrule
    \end{tabular}
\end{table}

\subsubsection{Circuit Logic}

Algorithm~\ref{alg:zksnark-semantic} outlines the logic for the \textbf{BuildCircuit} operation. It implements a deterministic matching logic to count True Positives. Note that to minimize circuit depth, the sorting of predictions (required for greedy matching) is performed off-circuit by the prover; the circuit only verifies that the provided sorting is correct and that the matching logic is valid.

\begin{algorithm}[!ht]
  \caption{\textbf{Circuit Definition: Model Auditing Logic}}
  \label{alg:zksnark-semantic}
  \textbf{Operation:} \textsc{BuildCircuit}$(\mathcal{L})$ \\
  \textbf{Description:} Defines constraints to verify precision and recall against a fixed ground truth $\mathcal{G}$.
  
  \begin{algorithmic}[1]
    \State \textbf{Signals Declaration:}
    \State \quad $\text{signal input private } \{\mathcal{D}, s_{sec}\}$ \Comment{Witness: Sorted Predictions}
    \State \quad $\text{signal input public } \{H(\mathcal{I}), \Delta_{commit}, \mathcal{G}, \tau_{prec}, \tau_{rec}, \dots\}$ (refer Table~\ref{tab:semanticinteroperability}) 
    \State \quad $\text{signal output public } \{\texttt{PASS}\}$

    \Statex \hrulefill

    \State \textbf{Step 1: Verify Dataset Integrity}
    \State \quad $\text{Assert}(H(\text{CircuitConstant}(\mathcal{I})) == \text{PublicInput}(H(\mathcal{I})))$

    \State \textbf{Step 2: Metric Verification Loop}
    \State \quad $\texttt{accumulation\_pass} \leftarrow 1$
    
    \For{$i$ $=$ $1$ to $N$} \Comment{Unrolled loop for fixed challenge set size}
        \State $TP \leftarrow 0$ 
        \State $TP_{crit} \leftarrow 0$
        
        \State \textit{// Optimization: Predictions $\mathcal{D}_i$ are supplied pre-sorted by confidence.}
        \State \textit{// Circuit strictly verifies the sort order and IoU matches.}
        
        \For{$j$ $=$ $1$ to $M_{max}$} \Comment{Loop up to max capacity}
            \State \textit{// [Greedy Matching Logic]}
            \State \textit{// If $d_{i,j}$ matches an unmatched $g_k$, increment counters}
             \If{$\text{IsMatch}(\mathcal{D}_{i,j}, \mathcal{G}_i)$}
                 \State $TP \leftarrow TP + 1$
                 \State $TP_{crit} \leftarrow TP_{crit} + c_{matched}$
             \EndIf
        \EndFor
        
        \State \Comment{Check Thresholds for Image $i$}
        \State $\texttt{pass}_i \leftarrow (TP \geq \tau_{prec} \cdot M_i) \land (TP \geq \tau_{rec} \cdot K_i)$
        \State $\texttt{crit\_pass}_i \leftarrow (TP_{crit} == \text{TotalCritical}_i)$
        
        \State $\texttt{accumulation\_pass} \leftarrow \texttt{accumulation\_pass} \land (\texttt{pass}_i \land \texttt{crit\_pass}_i)$
    \EndFor

    \State \textbf{Step 3: Enforce Public Outcome}
    \State \quad $\texttt{PASS} \ \text{<==} \ \texttt{accumulation\_pass}$ 

    \State \textbf{Step 4: Context Binding (Commitment)}
    \State \quad \textit{// Verify the Domain Tag matches the Model Audit Application ID (0x01)}
    \State \quad $\text{Assert}(\Delta_{commit} === \texttt{0x01010000})$  \Comment{Enforce correct App Context}
    \State \quad \textit{// Calculate commitment binding the context to the private predictions}
    \State \quad $c \leftarrow \text{PoseidonHash}(\Delta_{commit}; N; \mathcal{D}; T; \nu; s_{sec})$ \Comment{Enforce Binding and Hiding}
    
  \end{algorithmic}
\end{algorithm}

\subsubsection{Usage Scenarios}
The instantiation defined in Algorithm~\ref{alg:zksnark-semantic} supports several critical functions in the autonomous ecosystem:

\begin{itemize}
    \item \textbf{Regulatory Benchmarking:} An enforcement authority (e.g., NHTSA) can periodically issue challenge sets to vehicles. The vehicles can cryptographically prove they meet the current safety standards (e.g., ``Must detect 99 \text{ }\% of pedestrians") without needing to upload their proprietary software to a government server.
    \item \textbf{Marketplace Trust:} In a V2X data marketplace, a buyer can request a ``Proof of Quality" before purchasing semantic data from another vehicle. The seller proves their model is robust on a public benchmark, establishing trust in the data quality.
    \item \textbf{Privacy-Preserving Edge Case Collaboration:} When vehicles encounter rare or confusing scenarios (edge cases), they can prove to the fleet network that their model failed to detect an object (Recall failure) without revealing the specific incorrect bounding box, allowing the fleet to identify model weaknesses while maintaining privacy.
\end{itemize}

\section{Experiments and Performance Analysis} \label{sec:experiments}



In this section, we present the implementation of the proposed zk-SNARK framework for the applications presented in section~\ref{sec:sub:perception-integrity} and section~\ref{sec:sub:semantic}. We present the experiment and performance analysis on three different platforms that provide implementations of Groth16-based zero-knowledge proofs. SnarkJS~\cite{snarkjs} is written in JavaScript and web assembly. SnarkJS serves as a reference implementation that reflects the performance of a general-purpose, CPU-only Groth16 prover/verifier. Including this baseline allows us to quantify the performance gains obtained through low-level optimizations and GPU acceleration, which is critical for assessing feasibility in time-constrained AV settings. Rapidsnark~\cite{rapidsnark} is written in C++ and contains assembly level optimizations for Intel-based chips. To take full advantage of rapidsnark, it is required that dependencies for rapidsnark must also be compiled from source with chipset optimizations. Rapidsnark-GPU~\cite{rapidsnark-gpu} is GPU-based implementation of rapidsnark. In ideal conditions, it can achieve more than $4-10 X$ performance gain compared with its CPU-based counterpart. In our experiments we found that, due to the overhead introduced by memory operations, the performance gain is up to $2-4$ times only. 
We profile three resource intensive tasks: witness generation (see GenerateWitness from equation~\eqref{eq:generateproof}),  proof generation (see zkProof from equation~\eqref{eq:generateproof}) and proof verification (see zkVerify from equation~\eqref{eq:verifyproof}) in milliseconds (ms). All the performance analysis was done on an Intel Core Ultra 9 processor with 64 GB RAM and NVIDIA GeForce RTX 4070 with 8 GB memory. Rapidsnark was compiled and run using Clang compiler, while rapidsnark-gpu was complied and run using CUDA Toolkit (version $12.4$)\footnote{\url{https://developer.nvidia.com/cuda-12-4-0-download-archive}}. All implementations were executed 100 times, and the reported results correspond to the average over these runs.

\begin{figure}[!ht]
\caption{Performance of SnarkJS, RapidSnark and RapidSnark-GPU for proving perception integrity (refer section~\ref{sec:sub:perception-integrity})}
    \label{fig:exp:perception-integrity}
    \centering
    \includegraphics[width=0.75\linewidth]{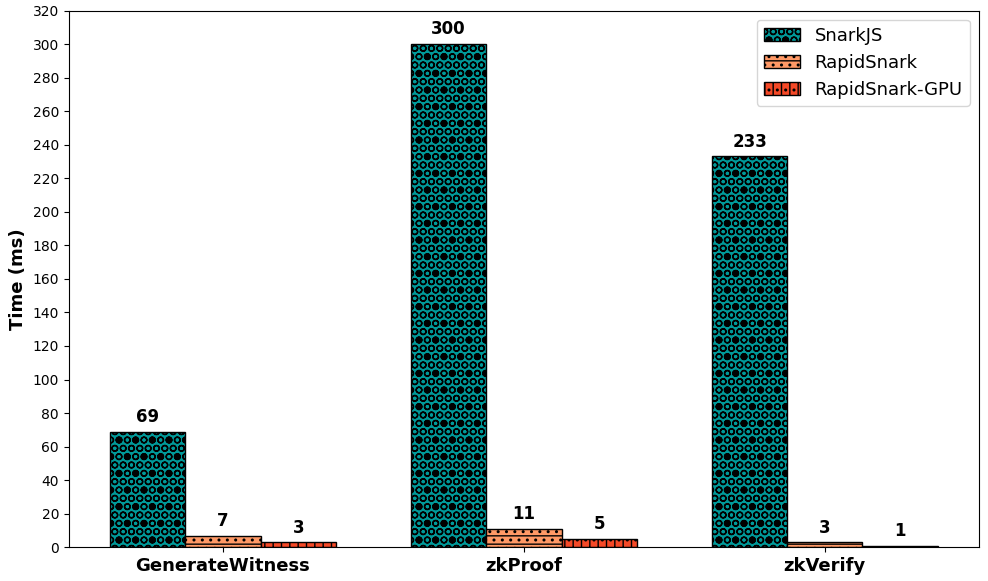}
\end{figure}

\begin{figure}[!ht]
\caption{Comparison of percentage of time taken by SnarkJS, RapidSnark and RapidSnark-GPU for each of the sub-routines.}
    \label{fig:exp:perception-integrity:subroutine}
    \centering
    \includegraphics[width=0.75\linewidth]{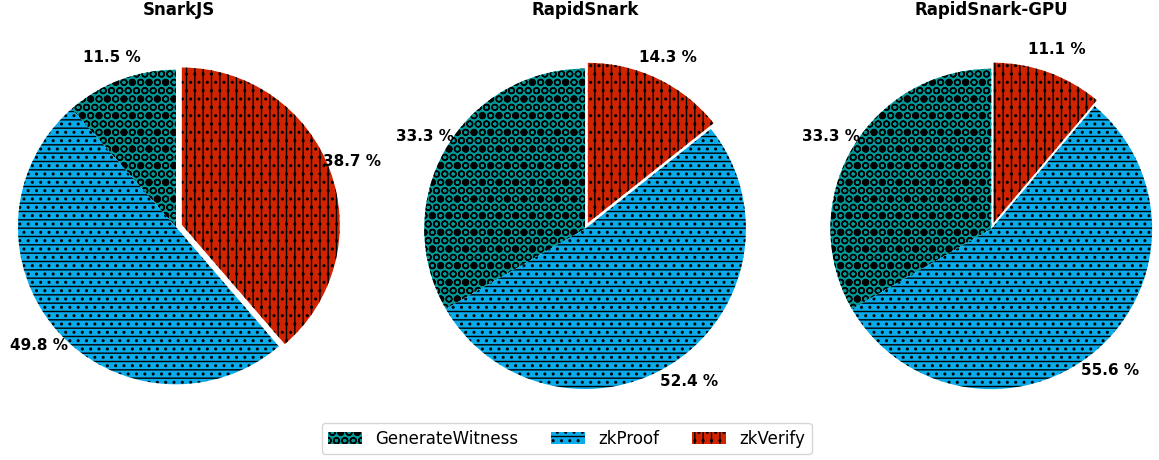}
\end{figure}
  
Figure~\ref{fig:exp:perception-integrity} shows the performance comparison for the proof of perception integrity (refer section~\ref{sec:sub:perception-integrity}) on three different platforms. There witness generation times are $69$ ms (SnarkJS), $7$ ms (rapidsnark) and $3$ ms (rapidsnark-gpu). This corresponds to an approximate $10 X$ speedup when moving from SnarkJS to the C++-based rapidsnark implementation, and an additional $2 X$ improvement when leveraging GPU acceleration. The observed gain is expected, since, SnarkJS is designed primarily for browser and JavaScript environments, whereas rapidsnark is implemented in optimized C++ with direct access to lower-level cryptographic primitives and hardware acceleration.

In the case of proof generation, SnarkJS takes $300$ ms while rapidsnark and rapidsnark-gpu takes $11$ ms and $5$ ms, respectively. In the case of proof generation, we see over $25 X$ performance gain between SnarkJS and rapidsnark. This is due to the fact that proof generation requires several cryptographic operations and hence, hardware acceleration using C++ outperforms by far JavaScript-based implementations. There is just over $2 X$ performance gain between CPU- and GPU-based implementations of rapidsnark. The disparity is most significant in proof verification. Rapidsnark verifies proofs in approximately $3$ ms, compared to $233$ ms for SnarkJS, yielding over a $75 X$ performance improvement. Rapidsnark-gpu reduces verification time further to $1$ ms, achieving an additional $3 X$ improvement over the CPU implementation. Although SnarkJS leverages WebAssembly for performance optimization, the experiments demonstrate that native C++ implementations with hardware support substantially outperform JavaScript-based approaches, particularly for elliptic-curve-heavy computations.

In figure~\ref{fig:exp:perception-integrity:subroutine}, the pie-chart shows the percentage of time taken by each of the sub-routines for the respective implementations. In case of SnarkJS, $49.8\text{ }\%$ is taken by witness generation, $38.7\text{ }\%$ by proof verification and $11.5\text{ }\%$ by proof verification. For rapidsnark, $33.3\text{ }\%$ is taken by witness generation, $52.4\text{ }\%$ is taken by proof generation and $14.3\text{ }\%$ is taken by proof verification. Rapidsnark-gpu sees a similar trend with $33.3 \text{ }\%$, $55.6\text{ }\%$ and $11.1\text{ }\%$ for witness generation, proof generation and proof verification, respectively.

Overall, the experimental results demonstrate that optimized C++ implementations, particularly when combined with GPU acceleration, provide substantial performance advantages for real-time perception integrity proofs in AVs.

\begin{figure}[!ht]
\caption{Performance of SnarkJS, RapidSnark and RapidSnark-GPU for model performance proof (refer section~\ref{sec:sub:semantic})}
    \label{fig:exp:semantic}
    \centering
    \includegraphics[width=0.75\linewidth]{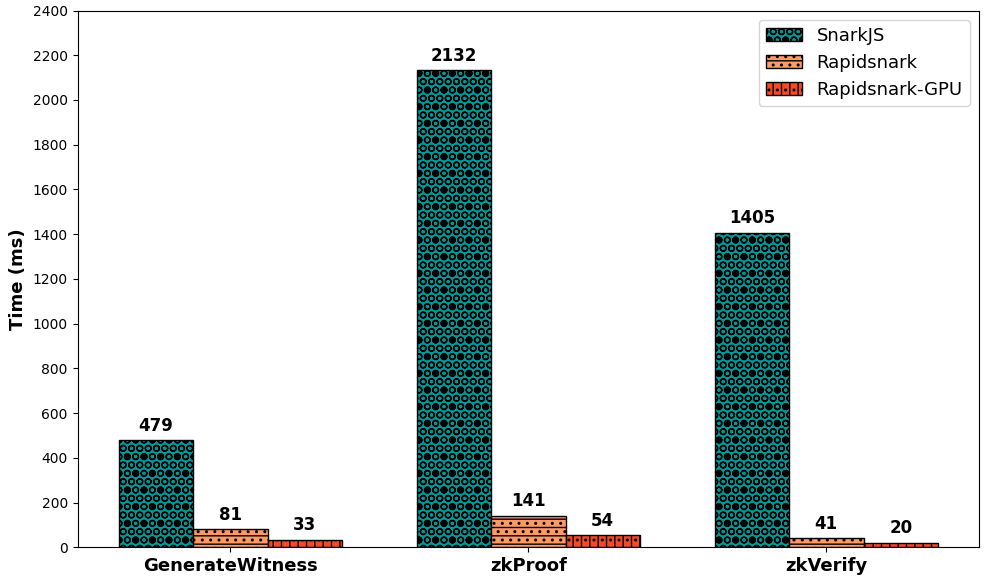}
\end{figure}

\begin{figure}[!ht]
\caption{Comparison of percentage of time taken by SnarkJS, RapidSnark and RapidSnark-GPU for each of the sub-routines.}
    \label{fig:exp:semantic:subroutine}
    \centering
    \includegraphics[width=0.75\linewidth]{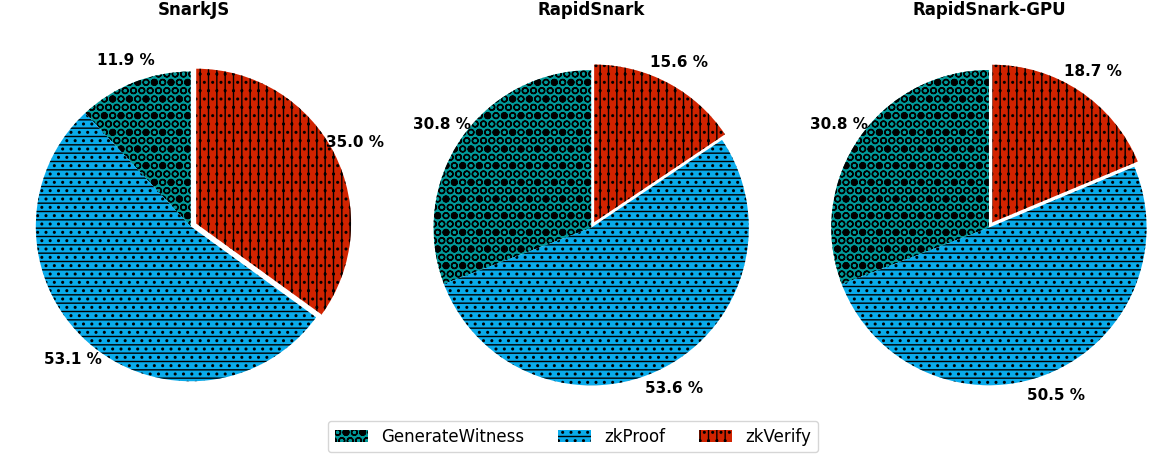}
\end{figure}

For demonstrating the performance analysis of the second application (refer section~\ref{sec:sub:semantic}), we utilize $5$ images and a total of $20$ ground truths. There were $5$ \textit{safety critical} objects and $20$ ground truths. The detection input provided by the model had $16$ detections. Out this $16$ detections $15$ mapped correctly with the ground truth  while $1$ was a phantom detection. The model also correctly detected $4$ out of $5$ \textit{safety critical} objects. Since, $100\text{ }\%$ of the detection for \textit{safety critical} objects is required for passing the model performance test, hence, the input for the detection provided by the vehicle resulted in failure or \texttt{Flag} set to $0$.
Figure~\ref{fig:exp:semantic} shows the comparative performance of SnarkJS, rapidsnark and rapidsnark-GPU for witness generation, proof generation and proof verification for the model performance evaluation or semantic understanding evaluation (refer section~\ref{sec:sub:semantic}). SnarkJS took $479$ ms, $2132$ ms and $1405$ ms for witness generation, proof generation and proof verification respectively. In contrast, rapidsnark took $81$ ms, $141$ ms and $41$ ms for the three operations while rapidsnark-GPU took $33$ ms, $54$ ms and $20$ ms respectively. 
Figure~\ref{fig:exp:semantic:subroutine} shows the percentage of time taken by each of the respective operations.  

\section{Limitations}
\label{sec:limitation}
In this section, we discuss the limitations of the proposed framework. 
\begin{itemize}
    \item Trust in EA: The framework assumes EA to correctly execute any required setup procedures and to securely handle any trapdoor information generated in the process. This assumption introduces a centralized trust dependency, which may raise concerns in large-scale or multi-stakeholder deployments. More broadly, in collaborative autonomous-vehicle settings involving multiple manufacturers (e.g., \textbf{\VASI}), participants may prefer not to rely on a single coordinating entity.
    While this work focuses on a trusted EA model, the framework is not inherently tied to this assumption.
    For example, proof systems such as PLONK, Halo2, and STARK-based constructions reduce or eliminate dependence on circuit-specific trusted setup through universal or transparent setup mechanisms. In such settings, the enforcing authority is responsible solely for declarative constraint specification, without being trusted for application-specific parameter generation.
    Alternatively, trust in the setup phase can be distributed across multiple independent stakeholders using multi-party computation (MPC) or multi-ceremony powers-of-tau constructions, thereby mitigating unilateral control and strengthening the overall trust model.  
    \item Complex Mathematical Operations: The framework works in a Galois field and hence, it may not be feasible to implement some of the mathematical operations, for example: division. Such constraints on the computation can potentially limit the applicability of the framework and may need workarounds or tricks like scaling to design the corresponding circuit.
    \item Proof Generation Time: The zero-knowledge proof is generated on the basis of the constraints imposed by EA. The proof generation time is proportional to the number of constraints; Hence, the framework limits the ability of EA to impose the constraints freely, especially for real-time applications.
    \item Proof of Execution (PoX): The perception stack of the \textit{ego} vehicle can act maliciously and provide favorable inputs to generate the zk-SNARK proof. Preventing such adversarial behavior with strong cryptographic guarantees requires a proof of execution. We discuss this further in section~\ref{sec:sub:pox}.
\end{itemize}

\section{Future Work}
\label{sec:future-work}
\subsection{Proof of Execution (PoX)}
\label{sec:sub:pox}
Hermes' Seal guarantees proof of computation on the outputs produced by the perception stack of the \textit{ego} vehicle. This design strategy allows efficient, privacy-preserving validation of key outputs without revealing internal model details or any other proprietary information. However, it assumes the integrity of the inputs to the proof system (e.g., model predictions in $\mathcal{I}$) and does not cryptographically enforce that the claimed model was actually executed on real sensor data. Hence, as a natural extension, future work can explore Proof of Execution (in short PoX) that can attest that both the model and inputs used to generate the proof are authentic and unmodified. PoX requires integration of  
trusted execution environments, thereby enhancing the overall trustworthiness of the perception pipeline without altering its modular structure. OP-TEE~\cite{op-tee} is one of the most popular trusted execution environment that provides isolation from general purpose linux kernel by utilizing ARM TrustZone~\cite{arm-trustzone}.

\subsection{Richer Constraints}
The framework can be extended to include complete semantic understanding, depth estimation, scene parsing, temporal proofs etc. Extension to these tasks require development and integration metrics that can unwind internal architecture of the models deployed on the perception stack. It is also worth noting that as computationally intensive task are integrated into the framework, the direction of decentralized infrastructure may also be explored.

\subsection{Integration with Leaderboards}
\label{sec:sub:leaderboard}
An exciting potential direction for future work is to integrate the zk-SNARK framework with the benchmark leaderboards, for example: KITTI~\cite{kitti-benchmark}, Waymo~\cite{waymo-dataset}, nuScenes~\cite{nuscenes-dataset}. This will enable verifiable submission of the results uploaded for evaluation by the models owners. By generating zero-knowledge proofs over model outputs, participants could prove compliance with evaluation metrics like accuracy, IoU etc., without revealing their model architecture, weights, or proprietary post processing algorithms. This promotes privacy preserving and fair evaluation process, and reduces the risk of tampering, leading to secure and trustworthy competition in sensitive or commercial domains. Such an integration has a potential to serve as a basis for certified model performance in regulated environments.

In section~\ref{sec:sub:semantic}, we introduced a metric called \textit{safety critical score}, designed to assess the performance of an autonomous vehicle’s perception stack with respect to objects of high criticality. This notion of criticality is modeled as a multivariate function that identifies objects of interest based on diverse environmental and adversarial contexts. For example, while a pedestrian may not always be considered critical, while a partially occluded pedestrian jaywalking near the vehicle represents a safety critical scenario. Despite its practical importance in the domain of autonomous vehicles, limited attention has been given to developing leaderboards that incorporate such metrics to evaluate the robustness of perception systems under distributional shifts. 

\section{Conclusion}
\label{sec:conclusion}
In this work we proposed Hermes' Seal: a modular, deployment ready framework that leverages zk-SNARKs to provide zero-knowledge proofs over perception outputs in autonomous systems thereby enabling \textbf{\VASI}. The verifiable yet privacy preserving assertions about object detections and semantic constraints enhances trust in AI-based decision making without exposing internal model logic or raw sensor data. Its model agnostic design allows seamless integration into existing perception stacks with minimal architectural changes, making it suitable for real-world deployment across heterogeneous platforms. This foundation opens the door to a range of future enhancements, including proof of execution and secure leaderboard submissions paving the way toward a verifiable and interoperable AI ecosystem for autonomous mobility.


\bibliographystyle{unsrt}  
\bibliography{references}  

\appendix
\section{Groth16 zk\textnormal{-}SNARK: Mathematical Overview}
\label{app:groth16}

\subsection{R1CS and QAP from Arithmetic Circuit}
Let $\mathbb{F}$ denote a finite field of prime order $p$, i.e.\ $\mathbb{F} = \mathbb{F}_p$.  Let the given computation be represented as an arithmetic circuit denoted by \texttt{Circuit}. 
All arithmetic in the \texttt{Circuit} and the proof system is performed over $\mathbb{F}$. 
\\
Let $n$ be the number of constraints in the \texttt{Circuit} and $\mathbf{w} \in \mathbb{F}^{m+1}$ be the solution vector of \texttt{Circuit}, .i.e.,  $\mathbf{w} = (w_0, w_1, \dots w_l, w_{l+1}, \dots w_m)^{\top}$, where $w_0= 1$\footnote{The first element $w_0$ is fixed to $1$ so that constant values can be included in the \texttt{Circuit}. This allows each constraint to express terms that do not depend on any input or intermediate variables.}. Then Rank-1 Constraint System (R1CS) is given by the following equation: 
\begin{equation}
\label{eq:appendix:r1cs}
    \forall i \in \{1, \dots, n\},\quad \big\langle \mathbf{a}_i,\mathbf{w}\big\rangle \cdot \big\langle \mathbf{b}_i,\mathbf{w}\big\rangle
    \;=\;
    \big\langle \mathbf{c}_i,\mathbf{w}\big\rangle,
\end{equation}
where $\mathbf{a}_i, \mathbf{b}_i, \mathbf{c}_i \in \mathbb{F}^{m+1}$ are row vectors that contains each of the constraint $i \in n$, and $\langle \cdot,\cdot\rangle$ denotes the standard inner product over the field $\mathbb{F}$ and each of the $i^{th}$ row defines exactly one multiplication gate in \texttt{Circuit}. Further, $\big\langle\mathbf{a}_i, \mathbf{w}\big\rangle = \sum_{j=0}^m A_{i,j}w_j$, is the left hand linear combination for the constraint $i$. Similarly, $\big\langle\mathbf{b}_i, \mathbf{w}\big\rangle = \sum_{j=0}^m B_{i,j}w_j$ is the right hand linear combination for the constraint $i$, while $\big\langle\mathbf{c}_i, \mathbf{w}\big\rangle = \sum_{j=0}^m C_{i,j}w_j$ is the output for the constraint $i$.
\\
Let $A,B,C\in\mathbb{F}^{n\times (m+1)}$ be the matrices whose $i^{th}$ rows are $\mathbf{a}_i,\mathbf{b}_i,\mathbf{c}_i$ respectively .i.e., 
\begin{equation}
    A = \begin{bmatrix}
        \mathbf{a}_1 \\
        \mathbf{a}_2 \\
        \vdots \\
        \mathbf{a}_n 
    \end{bmatrix},
    \qquad
    B = \begin{bmatrix}
        \mathbf{b}_1 \\
        \mathbf{b}_2 \\
        \vdots \\
        \mathbf{b}_n 
    \end{bmatrix},
    \qquad
    C = \begin{bmatrix}
        \mathbf{c}_1 \\
        \mathbf{c}_2 \\
        \vdots \\
        \mathbf{c}_n 
    \end{bmatrix}
\end{equation}
Then the R1CS is given by $A\mathbf{w} \circ B\mathbf{w} = C\mathbf{w}$, where $\circ$ denotes Hadamard product (element wise multiplication).
\\
To enable polynomial based cryptographic proofs, which are efficient, we convert the R1CS system to QAP (Quadratic Arithmaetic Polynomial). Langrage polynomial is used to convert R1CS to QAP.  
\\
Let $r_1, \dots, r_n \in \mathbb{F}$, then a Langrange basis for a polynomial of degree $\leq n$ is the set of polynomials $\{L_1(x), \dots L_n(x)\}$, where the values can be given using Kronecker delta .ie.., $L_i(r_j) = \delta_{ij}$, $\forall (i, j) \in n$. The degree of each of the polynomials in the set $\{L_1(x), \dots L_n(x)\}$ is $(n-1)$. Define three polynomials $A(x)$, $B(x)$ and $C(x)$ as follows:
\begin{equation}
\label{eq:lagrangepoly}
A(x) = \sum^{m}_{j=0} w_j\cdot A_j(x), \qquad B(x) = \sum^{m}_{j=0} w_j\cdot B_j(x), \qquad C(x) = \sum^{m}_{j=0} w_j\cdot C_j(x).
\end{equation}
where, $\forall$ $j \in \{0, m\}$, $A_j = \sum_{i=1}^n A_{i,j}L_i(x)$, $B_j = \sum_{i=1}^n B_{i,j}L_i(x)$, and $C_j = \sum_{i=1}^n C_{i,j}L_i(x)$. The polynomials of equation~\eqref{eq:lagrangepoly} can only satisfy the R1CS if $\forall i \in n$, $A(r_i)B(r_i) = C(r_i)$. $A(x)$, $B(x)$ and $C(x)$ must be evaluated for all the $n$ constraints. Hence, to convert evaluation of all the polynomials to a single divisibility (a.k.a succinct), define:
\begin{equation}
    \label{eq:succinct}
    A(x)B(x) - C(x) = H(x)t(x)
\end{equation}
where, $H(x) \in \mathbb{F}[m]$, while $t(x)$ is vanishing polynomial defined as follows:
\begin{equation}
\label{eq:vanish}
t(x)\;=\;\prod_{i=1}^n (x-r_i).
\end{equation}
Equation~\eqref{eq:succinct} provides the property of `succinct' to zk-SNARKs.

\subsection{Trusted Setup}
\label{app:groth16:srs}
Let $(\alpha, \beta, \gamma, \delta, \tau) \leftarrow \mathbb{F}^{5}$, be the secret trapdoor (toxic waste) selected by the Enforcement Authority EA. Then the
Structured Reference String (SRS) is defined as follows: $\{g^1_1, g^{\tau^1}_1, \dots, g^{\tau^d}_1\}$ and $\{g^1_2, g^{\tau^1}_2, \dots, g^{\tau^d}_2\}$, where $d < n$. Then the proving key $(pk)$ and verification key $(vk)$ is given by:
\begin{equation}
    \label{eq:pkvk}
    \begin{split}
        pk &= \Big(
        g_1,\, g_2,\, 
        \alpha G_1,\, \beta G_1,\, \beta G_2,\, 
        \delta G_1,\, \delta G_2, \\
        &\qquad \qquad \{ g_1^{A_j(\tau)} \}_{j=0}^{m},\;
         \{ g_2^{B_j(\tau)} \}_{j=0}^{m},\;
         \{ g_1^{C_j(\tau)} \}_{j=l+1}^{m},\\
        & \qquad \qquad \qquad \{ g_1^{\frac{\tau^{k} t(\tau)}{\delta}} \}_{k=0}^{n-2},\;
         \{ g_1^{\frac{L_j(\tau)}{\delta}} \}_{j=l+1}^{m}
    \Big) \\
    vk &= \Big(g_1,\, g_2,\, \alpha G_1,\, \beta G_2,\, \gamma G_2,\, \delta G_2,\, g_1^{\frac{L_i(\tau)}{\gamma}},\, \eta=( \alpha G_1, \beta G_2) \Big)
    \end{split}
\end{equation}
where, $\{ g_1^{A_j(\tau)} \}_{j=0}^{m} = \{\prod_{i=1}^{n} (g_1^{L_i(\tau)})^{A_{i,j}} | j = 0, \dots, m\}$, $\{ g_2^{B_j(\tau)} \}_{j=0}^{m} = \{\prod_{i=1}^{n} (g_2^{L_i(\tau)})^{B_{i,j}} | j = 0, \dots, m\}$ and $\{ g_1^{C_j(\tau)} \}_{j=l+1}^{m} = \{\prod_{i=1}^{n} (g_1^{L_i(\tau)})^{C_{i,j}} | j = (l+1), \dots, m\}$.

\subsection{zk-SNARK Proof}
The prover $\mathcal{P}$ samples two random field elements $(r, s) \xleftarrow{\$} \mathbb{F}$
to ensure zero-knowledge, and generates the proof components
$A \in \mathbb{G}_1$, $B \in \mathbb{G}_2$, and $C \in \mathbb{G}_1$ as follows:
\begin{equation}
    \begin{split}
        A &= (\alpha + A(\tau) + r\delta) \in \mathbb{G}_1 \\
        B &= (\beta + B(\tau) + s\delta) \in \mathbb{G}_2 \\
        C &= (C(\tau) + H(\tau)t(\tau) + rB(\tau) + sA(\tau) - rs\delta) \in \mathbb{G}_1. 
    \end{split}
\end{equation}
The zk-SNARK proof is given by: $\pi = (A, B, C)$, where, $A(\tau) = \prod_{j=0}^m (g_1^{A_j(\tau)})^{w_j}$, $B(\tau) = \prod_{j=0}^m (g_1^{B_j(\tau)})^{w_j}$ when computing under $\mathbb{G}_1$ and $B(\tau) = \prod_{j=0}^m (g_2^{B_j(\tau)})^{w_j}$ when computing under $\mathbb{G}_2$, $C(\tau) = \prod_{j=0}^m (g_1^{C_j(\tau)})^{w_j}$ and $H(\tau)t(\tau) = \sum_{k=0}^{n-2}h_k\tau^kt(\tau)$ since the degree of $H(x)$ is $(n-2)$ .i.e., $H(x) = \sum_{k=0}^{n-2}h_kx^k$ and $t(x)$ is the vanishing polynomial. Hence, $H(\tau)t(\tau)$ can be computed as follows: $\prod_{k=0}^{(n-2)}(g_1^{\frac{\tau^kt(\tau)}{\delta}})^{h_k}$.

\section{Security Game}
\label{appendix:security-game}

\begin{algorithm}[!h]
\caption{\textbf{:} Game Collision Resistance (CR): $\mathbf{G}_{H}^{\text{Coll}}(\mathcal{A})$}
\label{alg:cr-game}
    \begin{algorithmic}[1]
        \State $\{X_1, X_2\}\xleftarrow{\$} \mathcal{A}$ \Comment{$\{X_1, X_2\} \in X$}
        \State $Y_1 \leftarrow H(X_1)$ \Comment{$Y_1\in Y$}
        \State $Y_2 \leftarrow H(X_2)$ \Comment{$Y_2\in Y$}
        \If{$Y_{1} = Y_{2} $} \Comment{Condition for Collision}
            \State \textbf{return} $True$
        \Else
            \State \textbf{return} $False$
        \EndIf
\end{algorithmic}
\end{algorithm}

\begin{algorithm}[!h]
\caption{\textbf{:} Game BIND: $\mathbf{G}_{\text{CS}}^{\text{BIND}}(\mathcal{A})$, where CS=$(Commit, Open, Verify)$}
\label{alg:bind-game}
    \begin{algorithmic}[1]
        \State $\{(m_1, s_{sec_1}), (m_2, s_{sec_2})\}\xleftarrow{\$} \mathcal{A}$
        \State $c_1 \leftarrow Commit(m_1; s_{sec_1})$
        \State $c_2 \leftarrow Commit(m_2; s_{sec_2})$
        \If{$(m_1, s_{sec_1}) = (m_2, s_{sec_2})$ or $c_1 \neq c_2$}
            \State \textbf{return} $False$
        \Else
            \State \textbf{return} $True$
        \EndIf
\end{algorithmic}
\end{algorithm}

\begin{algorithm}[!h]
\caption{\textbf{:} Game HIDE: $\mathbf{G}_{\text{CS}}^{\text{HIDE}}(\mathcal{A})$, where CS=$(Commit, Open, Verify)$}
\label{alg:hide-game}
    \begin{algorithmic}[1]
            \State $\{(m_0, s_{sec_0}), (m_1, s_{sec_1}), state\}\xleftarrow{\$} \mathcal{A}$
            \State $b \xleftarrow{\$} \{0, 1\}$
            \State $c \leftarrow$ $Commit(m_b, s_{sec_b})$, where $c\leftarrow H(m_b; s_{sec_b})$ and $|H(m_b; s_{sec_b})| = n$
            \State $b' \xleftarrow{\$} \mathcal{A}((m_0, s_{sec_0}), (m_1, s_{sec_1}), state, c)$
            \State \textbf{return} $(b == b')$
\end{algorithmic}
\end{algorithm}

\subsection{Definitions and Security Notions}
\label{sec:definitions}
This section presents definitions and security notations that are used to map the parameters of AV domain to $\mathbb{F}_p$ and provide security analysis.  



\begin{definition}
\label{def:cr}
    \textbf{Collision Resistance:} Let $H: X \rightarrow Y$ be a hash function with $X$ as domain space and $Y$ as range space, where $X$ is the set of tuples of inputs and $Y$ is the set of outputs. For an adversary $\mathcal{A}$, we define the advantage of breaking the collision resistance of the function $H$ using Algorithm~\ref{alg:cr-game}, in the following way:
\end{definition}

\begin{equation}
    \text{Adv}_H^{\text{Coll}}(\mathcal{A})  = Pr[\mathbf{G}_{H}^{\text{Coll}}(\mathcal{A}) \mapsto 1].
\end{equation}

\begin{definition}
\label{def:commitment-scheme}
\textbf{Commitment Scheme:} Let $m \in \mathcal{M}$ and $s_{sec} \xleftarrow{\$} \mathcal{K}$, then a commitment scheme $CS$ is defined over three algorithms .i.e., $CS = (Commit, Open, Verify)$ as follows:
\begin{equation}
\begin{split}
    Commit(m, s_{sec}) &= H(m; s_{sec}) = c\\
    Open(c) &= (m, s_{sec}) \\
    Verify(c, m, s_{sec}) &= 
    \begin{cases}
        \top & \text{if } c = H(m; s_{sec})\\
        \perp & \text{otherwise}
    \end{cases}
\end{split}
\end{equation}
\end{definition}

\begin{definition}
    \label{def:bind}
\textbf{Binding Security of CS: } Let $CS$ be the commitment scheme defined in definition~\ref{def:commitment-scheme}. Then the binding security of the commitment scheme is given by algorithm~\ref{alg:bind-game}. For an adversary $\mathcal{A}$, the bind advantage on CS is given by the following equation:
\end{definition}
\begin{equation}
    \label{eq:cs-bind}
    \text{Adv}_{\text{CS}}^{\text{BIND}} (\mathcal{A}) = Pr[\mathbf{G}_{\text{CS}}^{\text{BIND}}(\mathcal{A}) \mapsto 1].
\end{equation}

\begin{definition}
    \label{def:hide}
\textbf{Hiding Security of CS: } Let $CS$ be  the commitment scheme defined in definition~\ref{def:commitment-scheme}. Then the hiding (indifferentiability) security of the commitment scheme CS is given by algorithm~\ref{alg:hide-game}. For an adversary $\mathcal{A}$, the hide advantage on CS is given by following equation:
\end{definition}
\begin{equation}
    \label{eq:cs-hide}
    \text{Adv}_{\text{CS}}^{\text{HIDE}} (\mathcal{A}) = Pr[\mathbf{G}_{\text{CS}}^{\text{HIDE}}(\mathcal{A}) \mapsto 1] - \frac{1}{2}.
\end{equation}


\end{document}